\documentclass[a4paper]{article}
\usepackage{amsfonts}
\usepackage{amssymb}
\usepackage{amsmath}
\usepackage{graphicx}
\usepackage{diagrams}
\usepackage{braket}
\usepackage{color}
\usepackage{arydshln}

\newtheorem{theorem}{Theorem}

\newenvironment{proof}[1][Proof]{\noindent\textbf{#1.} }{\ \rule{0.5em}{0.5em}}
\makeatletter
\def\@crosshairs{\vbox to0pt{}}
\makeatother
\input{tcilatex}

\newcommand{\untaggedfootnote}[1]{\let\thefootnote\relax\footnote{#1}}

\begin{document}

\title{Quantum Bounds for Option Prices}
\author{Paul McCloud \\
Department of Mathematics, University College London}
\maketitle

\begin{abstract}
Option pricing is the most elemental challenge of mathematical finance.
Knowledge of the prices of options at every strike is equivalent to knowing
the entire pricing distribution for a security, as derivatives contingent on
the security can be replicated using options. The available data may be
insufficient to determine this distribution precisely, however, and the
question arises: What are the bounds for the option price at a specified
strike, given the market-implied constraints?

Positivity of the price map imposed by the principle of no-arbitrage is here
utilised, via the Gelfand-Naimark-Segal construction, to transform the
problem into the domain of operator algebras. Optimisation in this larger
context is essentially geometric, and the outcome is simultaneously
super-optimal for all commutative subalgebras.

This generates an upper bound for the price of a basket option. With
innovative decomposition of the assets in the basket, the result is used to
create converging families of price bounds for vanilla options, interpolate
the volatility smile, price options on cross FX rates, and analyse the
relationships between swaption and caplet prices.\bigskip

\textbf{Keywords: }Option pricing; volatility smile; FX options; swaptions
and caplets; no-arbitrage principle; quantum probability; operator algebras;
Gelfand-Naimark-Segal construction.

\untaggedfootnote{Author email: p.mccloud@ucl.ac.uk}

\untaggedfootnote{Available on arXiv: arxiv.org/abs/1712.01385}

\untaggedfootnote{Available on SSRN: ssrn.com/abstract=3082561}
\end{abstract}

\newpage

\section{Introduction}

The incomplete market provides prices for a subset of the full universe of
securities, which constrains, but does not determine, the prices of
securities outside the mark-to-market subspace. In some cases, the available
data imposes model-independent limits on the possible prices for a security
that are sufficiently constrained as to provide useful guidelines for
pricing. These limits are the extremal valuations from the set of
arbitrage-free pricing models that satisfy the market constraints.

In this article, families of bounds for option prices are constructed from
finite-dimensional covariance matrices extracted from the price distribution
of the underlying assets. The approach allows for the arbitrary
decomposition of assets into sub-assets contingent on market events, a
property that is exploited to refine the bounds as more market information
is incorporated. The option price bound is then used to analyse problems
such as options on portfolios, interpolation of the Black-Scholes \cite%
{Black1973} implied volatility smile, the pricing of options on cross rates
in the foreign exchange market, and the relationships between swaption and
caplet prices.

Options on a security are a rich source of information regarding the pricing
measure, as the marginal distribution for the security is fully determined
from the prices of vanilla options. For the underlying security $\mathsf{a}$%
, integration by parts generates an expansion for the derived security $\phi
\lbrack \mathsf{a}]$ in terms of the vanilla put options $(k-\mathsf{a})^{+}$
and call options $(\mathsf{a}-k)^{+}$:%
\begin{align}
\phi \lbrack \mathsf{a}]=\phi \lbrack f]+\phi ^{\prime }[f](\mathsf{a}-f)& +%
\frac{1}{2}\int_{k=-\infty }^{f}\phi ^{\prime \prime }[k](k-\mathsf{a}%
)^{+}\,dk \\
& +\frac{1}{2}\int_{k=f}^{\infty }\phi ^{\prime \prime }[k](\mathsf{a}%
-k)^{+}\,dk  \notag
\end{align}%
Setting $f=\mathbb{E}[\mathsf{a}]$ to be the price of the underlying
security and taking the expectation in the pricing measure leads to the
Carr-Madan replication formula \cite{Carr1998} for the price of the derived
security:%
\begin{align}
\mathbb{E}[\phi \lbrack \mathsf{a}]]=\phi \lbrack f]& +\frac{1}{2}%
\int_{k=-\infty }^{f}\phi ^{\prime \prime }[k]\mathbb{E}[(k-\mathsf{a}%
)^{+}]\,dk \\
& +\frac{1}{2}\int_{k=f}^{\infty }\phi ^{\prime \prime }[k]\mathbb{E}[(%
\mathsf{a}-k)^{+}]\,dk  \notag
\end{align}%
If the price of the derived security is observed in the market, this formula
constrains the prices of vanilla options on the underlying security. The
challenge is to derive the bound for the price of the vanilla option subject
to the constraints imposed by the market prices of a finite collection of
derived securities. This bound should converge to a unique price for the
vanilla option as more market information is included.

The general problem considered here is the determination of bounds for the
price of a basket option:%
\begin{equation}
\mathbb{E}[(\sum_{n}\lambda _{n}\mathsf{a}_{n})^{+}]
\end{equation}%
for the assets $\mathsf{a}_{n}$, defined to be positive securities, and the
quantities $\lambda _{n}$, which may be positive or negative. The main
theoretical result of this article derives bounds for these options from the
matrix of moments:%
\begin{equation}
\mathbb{E}[\sqrt{\mathsf{a}_{m}\mathsf{a}_{n}}]
\end{equation}%
extracted from the price distribution. The diagonal elements of this matrix
are the prices of the assets, and the remaining elements are parametrised by
the volatilities and correlations of the square-roots of the assets. This
furnishes the bound with a convenient and intuitive parametrisation.

The method exploits the Gelfand-Naimark-Segal (GNS) construction \cite%
{Gelfand1943,Segal1947} to transform the problem into one involving operator
algebras, mirroring techniques applied in the study of quantum systems. This
foundational result from the theory of operator algebras is used to generate
an inner product on the securities $\mathsf{a}$ and $\mathsf{b}$:%
\begin{equation}
\braket{\mathsf{a}|\mathsf{b}}=\mathbb{E}[\mathsf{ab}]
\end{equation}%
thereby defining a Hilbert space structure on the securities. Aside from
technical details, the only properties required to validate this
construction are linearity and positivity of the price map, properties that
translate to the concepts of replicability and absence of arbitrage in the
finance application. The securities are represented as operators on the
Hilbert space -- the security $\mathsf{a}$ is represented as the operator $%
\hat{\mathsf{a}}$ whose action on the securities is defined by:%
\begin{equation}
\hat{\mathsf{a}}\ket{\mathsf{b}}=\ket{\mathsf{ab}}
\end{equation}%
Acting via pointwise multiplication, this representation identifies the
securities with the diagonal operators on the Hilbert space.

Optimisation within the wider context of all operators is essentially
geometric, allowing for the simple derivation of bounds for option prices.
The digital functions that indicate exercise in classical probability are
generalised as projections in quantum probability, and the central
theoretical result is the following observation for projections on a Hilbert
space.

\begin{theorem}
For the vectors $\ket{u_n}$ and the scalars $\lambda _{n}$, the supremum of
the valuation:%
\begin{equation}
\sum_{n}\lambda _{n}\braket{u_n|\mathsf{E}|u_n}
\end{equation}%
over all projections $\mathsf{E}$ is given by the sum of the positive
eigenvalues of a finite-dimensional self-adjoint matrix $P$ constructed from
the inner products $\braket{u_m|u_n}$ and the scalars $\lambda _{n}$.
\end{theorem}

\noindent From a practical perspective, the important element of this
theorem is the construction of the matrix $P$, which turns out to be a
simple application of standard matrix methods. The GNS construction
translates this theorem to the following result for options in
arbitrage-free pricing models.

\begin{theorem}
For the assets $\mathsf{a}_{n}$ and the quantities $\lambda _{n}$, the
option valuation:%
\begin{equation}
\mathbb{E}[(\sum_{n}\lambda _{n}\mathsf{a}_{n})^{+}]
\end{equation}%
is bounded above by the sum of the positive eigenvalues of a
finite-dimensional self-adjoint matrix $P$ constructed from the valuations $%
\mathbb{E}[\sqrt{\mathsf{a}_{m}\mathsf{a}_{n}}]$ and the quantities $\lambda
_{n}$.
\end{theorem}

\noindent Using creative decompositions of the assets, the bound in this
theorem is arbitrarily refined by extracting more information from the
market, generating families of volatility smiles that converge monotonically
to the market-implied smile.

By relying only on linearity and positivity of the map from security to
price, this approach is perfectly adapted to the economic principles of
replicability and the absence of arbitrage, so much so that the original
works by Gelfand and Naimark \cite{Gelfand1943} and Segal \cite{Segal1947}
could be considered as early results in the development of mathematical
finance. These results themselves emerged from the matrix approach to
quantum mechanics pioneered by Born, Jordan and Heisenberg \cite%
{Born1925,Born1926,Heisenberg1925}, through its formalisation in the work of
von Neumann \cite{Neumann1927,Neumann1930} and others, and are now a staple
in the study of operator algebras (see the standard texts \cite%
{Conway1990,Kadison1997,Kadison1997a}).

The precise correspondence with the principle of no-arbitrage encapsulated
in the GNS construction makes operator algebras the natural platform for
mathematical finance. The development is more commonly framed in the
familiar language of classical probability by taking the Arrow-Debreu
securities \cite{Arrow1954} as a basis for the market. While this approach
is largely unquestioned in the domain of mathematical finance, its validity
in the modelling of uncertainty has been the subject of debate in wider
economics circles, as this quote from Keynes \cite{Keynes1937} suggests.

\begin{quotation}
\noindent \emph{By \textquotedblleft uncertain" knowledge, let me explain, I
do not mean merely to distinguish what is known for certain from what is
only probable. \ldots\ About these matters there is no scientific basis on
which to form any calculable probability whatever. We simply do not know.
Nevertheless, the necessity for action and for decision compels us as
practical men to do our best to overlook this awkward fact and to behave
exactly as we should if we had behind us a good Benthamite calculation of a
series of prospective advantages and disadvantages, each multiplied by its
appropriate probability, waiting to be summed.}

John M. Keynes, 1937
\end{quotation}

\noindent As has been observed by economists such as Shackle \cite%
{Shackle1969}, the translation to classical probability is problematic as it
assumes that the range of outcomes indicated by the Arrow-Debreu securities
is known a priori, a requirement strangely at odds with the aims of
probabilistic modelling.

\begin{quotation}
\noindent \emph{We think of uncertainty as more than the existence in the
decision-maker's mind of plural and rival (mutually exclusive) hypotheses
amongst which he has insufficient epistemic grounds of choice. Decision, as
we mean the word, is creative and is able to be so through the freedom which
uncertainty gives for the creation of unpredictable hypotheses. Decision is
not choice amongst the delimited and prescribed moves in a game with fixed
rules and a known list of possible outcomes of any move or sequence of
moves. There is no assurance that any one can in advance say what set of
hypotheses a decision maker will entertain concerning any specified act
available to him. Decision is thought and not merely determinate response.}

George L. S. Shackle, 1969
\end{quotation}

\noindent Shackle rightly observes that \emph{`this language however is not
merely a vessel but a mould'} \cite{Shackle1972} that excludes the
possibility of surprise outcomes, though Shackle's attempts to remedy the
mathematics of classical probability are inadequate.

The solution is quantum probability. In the construction of Gelfand, Naimark
and Segal, the Arrow-Debreu securities manifest as commuting projections. A
fundamental result from the theory of operator algebras states that the
commutative algebra generated by these projections is unitarily isomorphic
to the bounded measurable functions on a measure space \cite%
{Neumann1949,Conway1990,Kadison1997}. In this perspective, the state space
is an emergent property of the market, naturally evolving as more potential
outcomes are uncovered. More important, though, is the corollary that the
algebra of all operators contains commutative subalgebras associated with
every possible configuration of the economy. Optimisations within the full
operator algebra are simultaneously super-optimal for all markets
represented as commutative subalgebras, and the analysis proceeds without
the need to make further assumptions on the nature of the economy.

While the resulting bounds for option prices could be determined using
purely classical methods, the ease with which they are derived using quantum
methods is noteworthy, and suggests further interesting applications. The
approach is liberated from the requirement that the securities form a
commutative algebra, leading to a framework for mathematical finance that
can be applied in noncommutative geometries \cite{McCloud2014,McCloud2017a}
with novel features not available to the classical variant.

\section{Bounds for option prices}

In this article, securities are identified with real-valued functions and
pricing models are identified with real-valued measures on the state space
of the economy. These identifications are based on the following core
assumptions:

\begin{itemize}
\item The security $\mathsf{a}$ is completely determined by specifying its
payoffs $\mathsf{a}[x]$ for each state $x$.

\item The pricing model $\mathsf{z}$ is completely determined by specifying
its prices $\mathsf{z}[X]$ of the Arrow-Debreu securities for each subset of
states $X$.
\end{itemize}

\noindent Appealing to the principle of replicability, the price of the
security $\mathsf{a}$ in the pricing model $\mathsf{z}$ is given by the
integral:%
\begin{equation}
\mathbb{E}[\mathsf{a}]=\int_{x}\mathsf{z}[dx]\,\mathsf{a}[x]
\end{equation}%
Prohibiting arbitrage then requires that the pricing measure is positive, so
that a security whose payoff is positive in all states of the economy has
positive price.

\subsection{The Gelfand-Naimark-Segal construction}

Positivity of the pricing model associated with a finite positive measure
enables the Gelfand-Naimark-Segal, or GNS, construction on the space $%
\mathsf{V}$ of securities. The content of the GNS construction is captured
in the statement that, for an arbitrage-free pricing model, the definition:%
\begin{equation}
\braket{\mathsf{a}|\mathsf{b}}=\mathbb{E}[\mathsf{a}^{\ast }\mathsf{b}]
\end{equation}%
provides an inner product on the securities $\mathsf{a},\mathsf{b}\in 
\mathsf{V}$. The security $\mathsf{a}\in \mathsf{V}$ is then represented as
a diagonal operator via pointwise-multiplication:%
\begin{equation}
\hat{\mathsf{a}}:\ket{\mathsf{b}}\in \mathsf{V}\mapsto \ket{\mathsf{ab}}\in 
\mathsf{V}
\end{equation}%
The apparent simplicity of this definition belies the technical challenges
of the construction, which needs to exclude securities with infinite prices
and factor out the degeneracy arising from securities whose payoffs are zero
almost everywhere. Standard results from the theory of operator algebras are
outlined here for completeness. The detail is not required for an
understanding of the finance applications that follow.

The foundational result is the Cauchy-Schwarz inequality \cite%
{Cauchy1821,Bouniakowsky1859,Schwarz1890} that positivity implies for
pricing.

\begin{theorem}[Cauchy-Schwarz inequality]
The arbitrage-free pricing model $\mathbb{E}$ satisfies the inequality:%
\begin{equation}
\left\vert \mathbb{E}[\mathsf{a}^{\ast }\mathsf{b}]\right\vert ^{2}\leq 
\mathbb{E}[\mathsf{a}^{\ast }\mathsf{a}]\mathbb{E}[\mathsf{b}^{\ast }\mathsf{%
b}]
\end{equation}%
for the securities $\mathsf{a},\mathsf{b}\in \mathsf{V}$.
\end{theorem}

\noindent Define the following subspaces of securities:%
\begin{align}
\mathsf{N}_{2}& =\{\mathsf{a}\in \mathsf{V}:\left\Vert \mathsf{a}\right\Vert
_{2}=0\} \\
\mathsf{V}_{2}& =\{\mathsf{a}\in \mathsf{V}:\left\Vert \mathsf{a}\right\Vert
_{2}<\infty \}  \notag
\end{align}%
where:%
\begin{equation}
\left\Vert \mathsf{a}\right\Vert _{2}=\sqrt{\mathbb{E}[\mathsf{a}^{\ast }%
\mathsf{a}]}
\end{equation}%
for the security $\mathsf{a}\in \mathsf{V}$. The first subspace includes the
securities that are zero almost everywhere, and the second subspace includes
the securities that are square-integrable, relative to the measure. The
pricing model is used to construct an inner product on the quotient space $%
\mathsf{V}_{2}/\mathsf{N}_{2}$. Denote by $\ket{\mathsf{a}}\equiv \mathsf{a}+%
\mathsf{N}_{2}$ the coset containing the security $\mathsf{a}\in \mathsf{V}%
_{2}$. The inner product of the two cosets $\ket{\mathsf{a}},\ket{\mathsf{b}}%
\in \mathsf{V}_{2}/\mathsf{N}_{2}$ is defined by:%
\begin{equation}
\braket{\mathsf{a}|\mathsf{b}}=\mathbb{E}[\mathsf{a}^{\ast }\mathsf{b}]
\end{equation}%
Repeated application of the Cauchy-Schwarz inequality demonstrates that this
is a well-defined inner product on the quotient space.

The topological completion of the quotient space is the Hilbert space:%
\begin{equation}
\mathsf{H}=\overline{\mathsf{V}_{2}/\mathsf{N}_{2}}
\end{equation}%
Define the following subspace of securities:%
\begin{equation}
\mathsf{V}_{\infty }=\{\mathsf{a}\in \mathsf{V}:\left\Vert \mathsf{a}%
\right\Vert _{\infty }<\infty \}
\end{equation}%
where:%
\begin{equation}
\left\Vert \mathsf{a}\right\Vert _{\infty }=\sup \{\sqrt{\mathbb{E}[\mathsf{b%
}^*\mathsf{a}^*\mathsf{a}\mathsf{b}]/\mathbb{E}[\mathsf{b}^*\mathsf{b}]}:%
\mathsf{b}\in \mathsf{V}_{2}\backslash \mathsf{N}_{2}\}
\end{equation}%
for the security $\mathsf{a}\in \mathsf{V}$. This subspace is closed under
the product, forming a subalgebra of the securities, and the GNS
construction represents the subalgebra as an algebra of bounded operators on
the Hilbert space.

\begin{theorem}[Gelfand-Naimark-Segal construction]
For the arbitrage-free pricing model $\mathbb{E}$, there is a representation:%
\begin{equation}
\mathsf{a}\in \mathsf{V}_{\infty }\mapsto \hat{\mathsf{a}}\in \mathcal{B}[%
\mathsf{H}]
\end{equation}%
of the securities as bounded operators on the Hilbert space $\mathsf{H}=%
\overline{\mathsf{V}_{2}/\mathsf{N}_{2}}$, such that the pricing model is a
pure state of the representation:%
\begin{equation}
\mathbb{E}[\mathsf{a}]=\braket{1|\hat{\mathsf{a}}|1}
\end{equation}%
for the security $\mathsf{a}\in \mathsf{V}_{\infty }$.
\end{theorem}

\noindent The representation in this construction is first defined on the
dense subspace $\mathsf{V}_{2}/\mathsf{N}_{2}\subset \mathsf{H}$ via
left-multiplication:%
\begin{equation}
\hat{\mathsf{a}}\ket{\mathsf{b}}=\ket{\mathsf{ab}}
\end{equation}%
for the securities $\mathsf{a}\in \mathsf{V}_{\infty }$ and $\mathsf{b}\in 
\mathsf{V}_{2}$, and extended to $\mathsf{H}$ by continuity, where
finiteness of the norm $\left\Vert \mathsf{a}\right\Vert _{\infty }$ ensures
that this extension is possible.

Heuristically, the left-multiplication operators are the diagonal operators
with respect to the basis of Dirac delta functions. This identification is
strictly valid only when the state space is discrete, but the analogy can be
a useful aid to understanding. The security is thus identified with a
diagonal operator on a Hilbert space, with the price of the security given
by the vacuum expectation of the operator.

By considering optimisation problems within the expanded domain of all
operators on the Hilbert space, it is possible to determine solutions that
are super-optimal for the restricted application. This can be used to derive
bounds for option prices.

\subsection{Super-optimal exercise strategies}

For the assets $\mathsf{a}_{n}$, defined to be positive securities, and the
quantities $\lambda _{n}$, which may be positive or negative scalars,
consider the option to receive the portfolio $\sum_{n}\lambda _{n}\mathsf{a}%
_{n}$. Exercise of the option is indicated by the Arrow-Debreu security $%
\mathsf{e}$, restricted so that it only takes the values zero or one, $\func{%
spec}[\mathsf{e}]\subset \{0,1\}$. The price of the option is then:%
\begin{equation}
p[\mathsf{e}]=\mathbb{E}[(\sum_{n}\lambda _{n}\mathsf{a}_{n})\mathsf{e}]
\end{equation}%
Optimal exercise happens when the option price is maximised over all
possible exercise strategies. In this case, optimal exercise corresponds to
the indicator $\mathsf{e}=(\sum_{n}\lambda _{n}\mathsf{a}_{n}\geq 0)$, with
option price:%
\begin{equation}
p=\mathbb{E}[(\sum_{n}\lambda _{n}\mathsf{a}_{n})^{+}]
\end{equation}%
The option price is obtained as the supremum price over a range of
securities, each identified by its exercise strategy. Without additional
information regarding the measure, it is not possible to refine this
statement. It is possible, however, to obtain a super-optimal price for the
option that requires only partial information from the pricing model.

Using the GNS construction associated with the pricing model, the option
price is expressed as:%
\begin{equation}
p[\mathsf{e}]=\sum_{n}\lambda _{n}\braket{\sqrt{\mathsf{a}_n}|\hat{%
\mathsf{e}}|\sqrt{\mathsf{a}_n}}
\end{equation}%
The optimal option price is the supremum of this expression over projections 
$\hat{\mathsf{e}}$ in the subalgebra of \emph{left-multiplication}
operators. This is bounded above by the supremum of the expression:%
\begin{equation}
p[\mathsf{E}]=\sum_{n}\lambda _{n}\braket{\sqrt{\mathsf{a}_n}|\mathsf{E}|%
\sqrt{\mathsf{a}_n}}
\end{equation}%
over projections $\mathsf{E}$ in the algebra of \emph{all} operators. The
beauty of this observation is that the evaluation of the supremum over all
projections is essentially geometric, requiring optimisation only over the
projections on the finite-dimensional subspace spanned by the cosets $%
\ket{\sqrt{\mathsf{a}_n}}$ associated with the square-roots of the assets.

\subsection{Eigenvalue solution for the supremum}

Motivated by the preceding argument, consider the following problem on a
Hilbert space $\mathsf{H}$: Given the vectors $\ket{u_n}$ and the scalars $%
\lambda _{n}$, determine the supremum of the valuations $\sum_{n}\lambda _{n}%
\braket{u_n|\mathsf{E}|u_n}$ over all projections $\mathsf{E}$. This
supremum is determined in the following theorem.

\begin{theorem}
Let $\mathsf{H}$ be a Hilbert space. For the vectors $\ket{u_n}\in \mathsf{H}
$ and the scalars $\lambda _{n}\in \mathbb{R}$, define the
finite-dimensional matrices $Q$ and $\Lambda $ with matrix elements:%
\begin{align}
Q_{mn}& =\braket{u_m|u_n} \\
\Lambda _{mn}& =\lambda _{n}\delta _{mn}  \notag
\end{align}%
For a decomposition $Q=S^{\ast }S$ of the positive semi-definite matrix $Q$
in terms of a matrix $S$, define the self-adjoint matrix $P=S\Lambda S^{\ast
}$. Then the supremum of the valuation:%
\begin{equation}
\sum_{n}\lambda _{n}\braket{u_n|\mathsf{E}|u_n}
\end{equation}%
over all projections $\mathsf{E}\in \mathcal{B}[\mathsf{H}]$ is given by the
sum of the positive eigenvalues of the matrix $P$.
\end{theorem}

\begin{proof}
The aim is to determine the supremum:%
\begin{align}
p=\sup \{& \sum_{n}\lambda _{n}\braket{u_n|\mathsf{E}|u_n}: \\
& \mathsf{E}\in \mathcal{B}[\mathsf{H}],\mathsf{E}^{\ast }=\mathsf{E},\func{%
spec}[\mathsf{E}]\subset \{0,1\}\}  \notag
\end{align}%
The problem is simplified by decomposing the Hilbert space, $\mathsf{H}=%
\mathsf{H}_{0}\oplus \mathsf{H}_{1}$, where $\mathsf{H}_{0}$ is the
finite-dimensional Hilbert space spanned by the vectors and $\mathsf{H}_{1}$
is its orthogonal complement in $\mathsf{H}$. The valuation depends only on
the restriction of the projection to the subspace:%
\begin{equation}
\sum_{n}\lambda _{n}\braket{u_n|\mathsf{E}|u_n}=\sum_{n}\lambda _{n}%
\braket{u_n|\mathsf{E}_0|u_n}
\end{equation}%
where the projection is decomposed relative to the decomposition of the
Hilbert space:%
\begin{equation}
\mathsf{E}=%
\begin{bmatrix}
\mathsf{E}_{0} & \mathsf{F} \\ 
\mathsf{F}^{\ast } & \mathsf{E}_{1}%
\end{bmatrix}%
\end{equation}%
for the operators $\mathsf{E}_{0}\in \mathcal{B}[\mathsf{H}_{0}]$, $\mathsf{E%
}_{1}\in \mathcal{B}[\mathsf{H}_{1}]$ and $\mathsf{F}\in \mathcal{B}[\mathsf{%
H}_{1},\mathsf{H}_{0}]$. The upper-left operator $\mathsf{E}_{0}$ is
self-adjoint, but it is not necessarily a projection. Instead, the
projection condition $\mathsf{E}^{2}=\mathsf{E}$ translates to the property:%
\begin{equation}
\mathsf{E}_{0}(1-\mathsf{E}_{0})=\mathsf{FF}^{\ast }
\end{equation}%
Interference from the off-diagonal operator $\mathsf{F}$ prevents $\mathsf{E}%
_{0}$ from being a projection. The operator $\mathsf{FF}^{\ast }$ is
positive semi-definite, so the projection property implies that $\func{spec}[%
\mathsf{E}_{0}]\subset \lbrack 0,1]$ with interference creating the
possibility of eigenvalues between zero and one.

The restricted projection $\mathsf{E}_{0}\in \mathcal{B}[\mathsf{H}_{0}]$ is
diagonalised as:%
\begin{equation}
\mathsf{E}_{0}=\sum_{i}\omega _{i}\ket{z_i}\bra{z_i}
\end{equation}%
where $\ket{z_i}\in \mathsf{H}_{0}$ are diagonalising orthonormal basis
eigenvectors and the eigenvalues $\omega _{i}\in \mathbb{R}$ satisfy $0\leq
\omega _{i}\leq 1$. Using this diagonalisation, the valuation becomes:%
\begin{equation}
\sum_{n}\lambda _{n}\braket{u_n|\mathsf{E}_0|u_n}=\sum_{i}\omega _{i}%
\braket{z_i|\mathsf{P}|z_i}
\end{equation}%
where the self-adjoint operator $\mathsf{P}\in \mathcal{B}[\mathsf{H}_{0}]$
is defined by:%
\begin{equation}
\mathsf{P}=\sum_{n}\lambda _{n}\ket{u_n}\bra{u_n}
\end{equation}%
Among the restricted projections that share the eigenvectors $\ket{z_i}$,
the maximum valuation is obtained by using the projection onto the subspace
of $\mathsf{H}_{0}$ spanned by the eigenvectors for which the diagonal
element $\braket{z_i|\mathsf{P}|z_i}$ is positive. The expression for the
supremum is then:%
\begin{align}
p=\sup \{& \sum_{i}\braket{z_i|\mathsf{P}|z_i}^{+}: \\
& \ket{z_i}\in \mathsf{H}_{0}\text{ orthonormal basis}\}  \notag
\end{align}

The valuation $\sum_{i}\braket{z_i|\mathsf{P}|z_i}^{+}$ is the sum of the
positive diagonal elements of the operator $\mathsf{P}$. A straightforward
appeal to the Schur-Horn theorem \cite{Schur1923,Horn1954} demonstrates that
this is bounded above by the sum of the positive eigenvalues of $\mathsf{P}$%
. To see this, first assume without loss of generality that the diagonal
elements $\braket{z_i|\mathsf{P}|z_i}$ and the eigenvalues $p_{i}$ of $%
\mathsf{P}$ are arranged in non-increasing order. The Schur-Horn theorem
states that:%
\begin{equation}
\sum_{i=1}^{j}\braket{z_i|\mathsf{P}|z_i}\leq \sum_{i=1}^{j}p_{i}
\end{equation}%
for all $j$. Taking the maximum over $j$, first on the right and then on the
left, shows that:%
\begin{equation}
\max_{j}[\sum_{i=1}^{j}\braket{z_i|\mathsf{P}|z_i}]\leq
\max_{j}[\sum_{i=1}^{j}p_{i}]
\end{equation}%
The required result then follows from the observation that the maxima in
this expression are given by the sum of the positive elements in their
respective sequences, so that:%
\begin{equation}
\sum_{i}\braket{z_i|\mathsf{P}|z_i}^{+}\leq \sum_{i}p_{i}^{+}
\end{equation}

The sum of the positive eigenvalues of $\mathsf{P}$ bounds the sum of the
positive diagonal elements of $\mathsf{P}$, and so provides an upper bound
for the supremum. This bound is attained by using the projection onto the
subspace of $\mathsf{H}_{0}$ spanned by the eigenvectors of $\mathsf{P}$
with positive eigenvalues. The supremum of the valuations is then finally
identified with the sum of the positive eigenvalues of $\mathsf{P}$:%
\begin{equation}
p=\sum_{i}p_{i}^{+}
\end{equation}%
The supremum is thus related to the solution of a finite-dimensional
eigenvalue problem, and is obtained as the sum of the positive roots of a
polynomial of order matching the dimension of the subspace spanned by the
vectors.

The computation of the eigenvalues is enabled by expressing the problem in
terms of an orthonormal basis $\ket{z_i}\in \mathsf{H}_{0}$ for the
subspace. The algorithm seeks to construct the matrix $P=[P_{ij}]$ from the
input matrix $Q=[Q_{mn}]$, where the matrix elements are:%
\begin{align}
P_{ij}& =\braket{z_i|\mathsf{P}|z_j} \\
Q_{mn}& =\braket{u_m|u_n}  \notag
\end{align}%
The solution requires the matrices $\Lambda =[\Lambda _{mn}]$ and $%
S=[S_{in}] $ with matrix elements:%
\begin{align}
\Lambda _{mn}& =\lambda _{n}\delta _{mn} \\
S_{in}& =\braket{z_i|u_n}  \notag
\end{align}%
The essential relationships among these matrices are:%
\begin{align}
P& =S\Lambda S^{\ast } \\
Q& =S^{\ast }S  \notag
\end{align}%
The program for solving the eigenvalue problem is now clear: First decompose
the positive semi-definite matrix $Q$ in the form $S^{\ast }S$, then solve
for the eigenvalues of the self-adjoint matrix $P=S\Lambda S^{\ast }$. Any
such decomposition for the matrix $Q$ generates the same result, as the
eigenvalue problem is unaffected by unitary transformations. The solution
thus depends only on the scalars $\lambda _{n}$ and the inner products $%
\braket{u_m|u_n}$, and the dimension of the eigenvalue problem is the rank
of the matrix with elements given by these inner products.
\end{proof}

\section{Applications of the option price bound}

The GNS\ construction determines an inner product on the securities, and
this relates the result of the previous section to the prices of options.

\begin{theorem}
For the arbitrage-free pricing model $\mathbb{E}$, the assets $\mathsf{a}%
_{n} $ and the quantities $\lambda _{n}$, define the finite-dimensional
matrices $Q$ and $\Lambda $ with matrix elements:%
\begin{align}
Q_{mn}& =\mathbb{E}[\sqrt{\mathsf{a}_{m}\mathsf{a}_{n}}] \\
\Lambda _{mn}& =\lambda _{n}\delta _{mn}  \notag
\end{align}%
For a decomposition $Q=S^{\ast }S$ of the positive semi-definite matrix $Q$
in terms of a matrix $S$, define the self-adjoint matrix $P=S\Lambda S^{\ast
}$. Then the option valuation:%
\begin{equation}
\mathbb{E}[(\sum_{n}\lambda _{n}\mathsf{a}_{n})^{+}]
\end{equation}%
is bounded above by the sum of the positive eigenvalues of the matrix $P$.
\end{theorem}

\begin{proof}
The GNS construction translates this statement into the language of the
previous theorem. The proof is then completed by observing that the option
valuation takes the form:%
\begin{equation}
\sum_{n}\lambda _{n}\braket{\sqrt{\mathsf{a}_n}|\hat{\mathsf{e}}|\sqrt{%
\mathsf{a}_n}}
\end{equation}%
where the projection $\hat{\mathsf{e}}$ is the left-multiplication operator
associated with the digital security $\mathsf{e}$ indicating the exercise
strategy for the option. The option valuation is thus bounded above by the
supremum of this expression over all projections which, by the previous
theorem, is the sum of the positive eigenvalues of the matrix $P$.
\end{proof}

The matrix $P$ is constructed from the diagonal matrix $\Lambda $, whose
diagonal elements are the quantities of the portfolio, and the symmetric
matrix $Q$, whose elements are the moments $\mathbb{E}[\sqrt{\mathsf{a}_m%
\mathsf{a}_n}]$ of the measure. The diagonal elements of $Q$ are the prices
of the assets, typically sourced from available market data. The
off-diagonal elements introduce additional volatility and correlation
dependencies, providing the model parametrisation for the bound.

The theorem generates a bound on the price of the basket option. In this
application, the matrices $Q$ and $\Lambda $ are given by:%
\begin{align}
Q& =%
\begin{bmatrix}
\sqrt{f_{m}f_{n}}\,q_{mn}%
\end{bmatrix}
\\
\Lambda & =%
\begin{bmatrix}
\lambda _{n}\delta _{mn}%
\end{bmatrix}
\notag
\end{align}%
Here, $f_{n}$ is the price of the $n$th asset and $q_{mn}$ is the normalised
cross-term for the $m$th and $n$th assets:%
\begin{align}
f_{n}& =\mathbb{E}[\mathsf{a}_n] \\
q_{mn}& =\frac{\mathbb{E}[\sqrt{\mathsf{a}_m\mathsf{a}_n}]}{\sqrt{\mathbb{E}[%
\mathsf{a}_m]\mathbb{E}[\mathsf{a}_n]}}  \notag
\end{align}%
The cross-term is driven by the volatilities of the assets and the
correlation between them, and can be expressed as:%
\begin{equation}
q_{mn}=\sqrt{(1-\nu _{m})(1-\nu _{n})}+\rho _{mn}\sqrt{\nu _{m}\nu _{n}}
\end{equation}%
where $\nu _{n}$ is the normalised variance of the square-root of the $n$th
asset and $\rho _{mn}$ is the correlation between the square-roots of the $m$%
th and $n$th assets:%
\begin{align}
\nu _{n}& =\frac{\mathbb{E}[\mathsf{a}_n]-\mathbb{E}[\sqrt{\mathsf{a}_n}]^{2}%
}{\mathbb{E}[\mathsf{a}_n]} \\
\rho _{mn}& =\frac{\mathbb{E}[\sqrt{\mathsf{a}_m\mathsf{a}_n}]-\mathbb{E}[%
\sqrt{\mathsf{a}_m}]\mathbb{E}[\sqrt{\mathsf{a}_n}]}{\sqrt{(\mathbb{E}[%
\mathsf{a}_m]-\mathbb{E}[\sqrt{\mathsf{a}_m}]^{2})(\mathbb{E}[\mathsf{a}_n]-%
\mathbb{E}[\sqrt{\mathsf{a}_n}]^{2})}}  \notag
\end{align}%
The price is positive, $f_{n}>0$, the root-variance lies in the range $0\leq
\nu _{n}\leq 1$, and the correlation lies in the range $-1\leq \rho
_{mn}\leq 1$. This completes the parametrisation of the model.

This result stands alone as an interesting application, providing an
intuitive parametrisation for the price of the basket option. There is an
ingenious interpretation of the result that extends its applicability beyond
basket options, leading to a significant family of upper bounds that
converges to the exact price as more information is absorbed. The key is to
recognise that the decomposition of the portfolio into constituent assets
can be arbitrarily refined, with each such decomposition yielding a new
upper bound.

The range of results obtained in this manner is limited only by the
creativity applied in the deconstruction of the portfolio. Taking a
partition of unity constructed from vanilla call and put options generates a
convergent family of upper bounds for the volatility smile. Another
application for interest rate products derives from the decomposition of the
swap rate in terms of its constituent forward rates, creating links between
the prices of swaptions and caplets. These applications are explored below.

\subsection{Vanilla options}

\begin{figure}[p]
\begin{tabular}{c}
\includegraphics[width=0.95%
\linewidth]{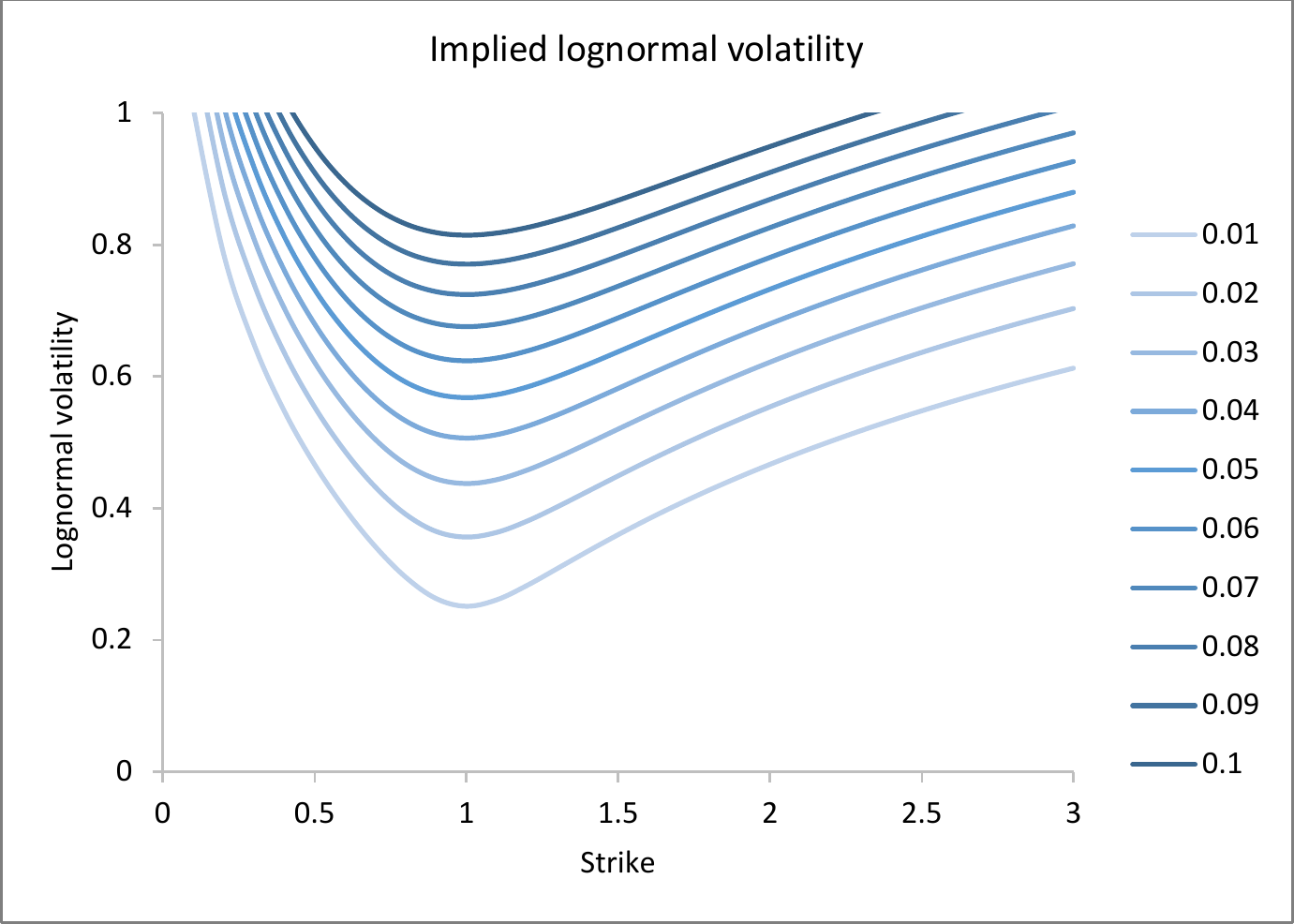}
\\ 
\includegraphics[width=0.95%
\linewidth]{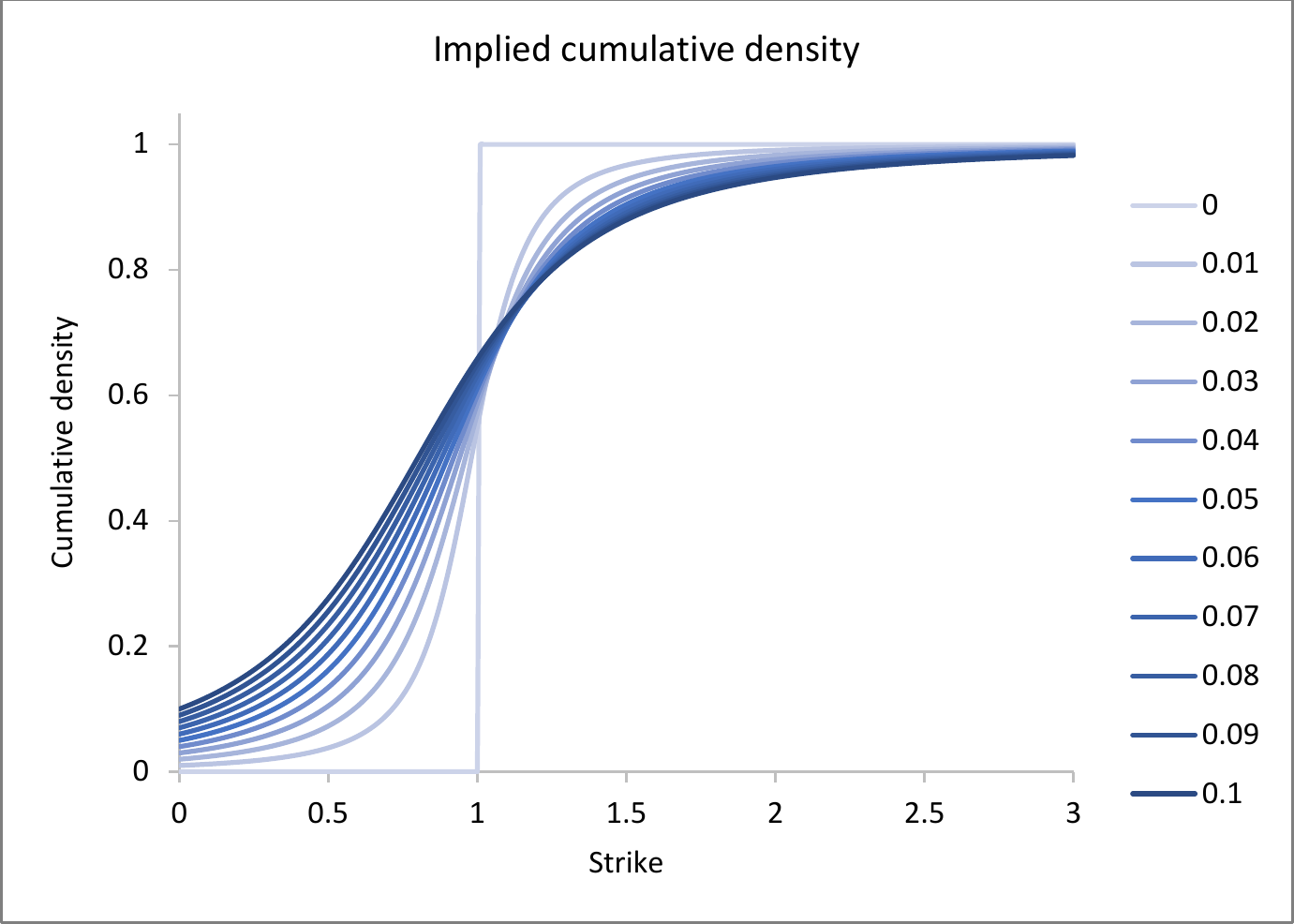}%
\end{tabular}%
\caption{The upper bound for the vanilla option price. In these graphs, the
price of the asset is fixed at 1 and the root-variance takes a range of
values between 0 and 0.1. The first graph expresses the bound in terms of
the implied lognormal volatility that recreates the option price in the
Black-Scholes model. The second graph shows the cumulative density function
implied by the option price, generated by differentiating the bound with
respect to the strike. The density combines a point density at strike 0 with
probability given by the root-variance, and a continuous density supported
on the upper half-line.}
\end{figure}

The eigenvalue problem as formulated above is solved using standard matrix
methods. In the case of two assets, this reduces to a quadratic equation
with an explicit solution. For the asset $\mathsf{a}$ and positive strike $k$%
, the price of the option to receive the portfolio $\mathsf{a}-k$ is bounded
above by:%
\begin{equation}
\mathbb{E}[(\mathsf{a}-k)^+]\leq p_{-}^{+}+p_{+}^{+}
\end{equation}%
where $p_{-}$ and $p_{+}$ are the eigenvalues of the matrix $P$ constructed
from the diagonal matrix $\Lambda $, whose diagonal elements depend on the
strike, and the symmetric matrix $Q$, whose elements are generated from the
moments $\mathbb{E}[\sqrt{\mathsf{a}}]$ and $\mathbb{E}[\mathsf{a}]$ of the
measure. The diagonal element of the matrix $Q$ is the price of the asset,
which is marked to market. The off-diagonal element introduces an additional
volatility dependency in the bound, controlled by a single model parameter.

The result is applied to generate a bound on the price of the vanilla
option. In this application, the matrices $Q$ and $\Lambda $ are given by:%
\begin{align}
Q& =%
\begin{bmatrix}
f & \sqrt{f(1-\nu )} \\ 
\sqrt{f(1-\nu )} & 1%
\end{bmatrix}
\\
\Lambda & =%
\begin{bmatrix}
1 & 0 \\ 
0 & -k%
\end{bmatrix}
\notag
\end{align}%
Here, $f$ is the price of the asset and $\nu $ is the normalised variance of
the square-root of the asset:%
\begin{align}
f& =\mathbb{E}[\mathsf{a}] \\
\nu & =\frac{\mathbb{E}[\mathsf{a}]-\mathbb{E}[\sqrt{\mathsf{a}}]^{2}}{%
\mathbb{E}[\mathsf{a}]}  \notag
\end{align}%
The price is positive, $f>0$, and the root-variance lies in the range $0\leq
\nu \leq 1$.

Represent the matrix $Q$ as $S^{\ast }S$, where $S$ is the lower-triangular
matrix generated using the Cholesky decomposition:%
\begin{equation}
S=%
\begin{bmatrix}
\sqrt{f\nu } & 0 \\ 
\sqrt{f(1-\nu )} & 1%
\end{bmatrix}%
\end{equation}%
The eigenvalue solution for the supremum is derived from the matrix $P$
defined as the combination $S\Lambda S^{\ast }$ of the diagonal matrix $%
\Lambda $ with the lower-triangular matrix $S$:%
\begin{equation}
P=%
\begin{bmatrix}
f\nu & f\sqrt{\nu (1-\nu )} \\ 
f\sqrt{\nu (1-\nu )} & f(1-\nu )-k%
\end{bmatrix}%
\end{equation}%
The eigenvalues $p$ of the matrix $P$ are the solutions of the quadratic
equation derived from the determinant condition $\det [P-p]=0$:%
\begin{equation}
p^{2}-(f-k)p-fk\nu =0
\end{equation}%
There are two solutions to this quadratic equation, but only one of them is
positive. This eigenvalue provides the bound for the option price:%
\begin{equation}
\mathbb{E}[(\mathsf{a}-k)^+]\leq \frac{1}{2}(f-k)+\frac{1}{2}\sqrt{%
(f-k)^{2}+4fk\nu }
\end{equation}%
The bound extracts only two moments, the price and root-variance, from the
measure, and applies to all pricing models calibrated to these moments.

\subsection{Refining the option price bound}

The bound for the option price is refined by using a partition of unity to
decompose the option payoff, resulting in a bound that is constrained by the
prices of options at a finite set of strikes. Consider the partition assets $%
u_{n}[\mathsf{a}]$, satisfying the properties $u_{n}[\mathsf{a}]\geq 0$ and $%
\sum_{n}u_{n}[\mathsf{a}]=1$. Using this partition, the spread between the
asset $\mathsf{a}$ and strike $k$ is expressed as the portfolio:%
\begin{equation}
\mathsf{a}-k=\sum_{n}\mathsf{a}u_{n}[\mathsf{a}]-k\sum_{n}u_{n}[\mathsf{a}]
\end{equation}%
This decomposition generates a bound for the option price from a matrix
whose diagonal elements are the prices of the partition assets scaled by the
asset and the strike. The utility of the bound then depends on whether an
intuitive parametrisation can be found for the off-diagonal moments implied
by the partition.

\begin{figure}[t]
\begin{tabular}{c}
\includegraphics[width=0.95%
\linewidth]{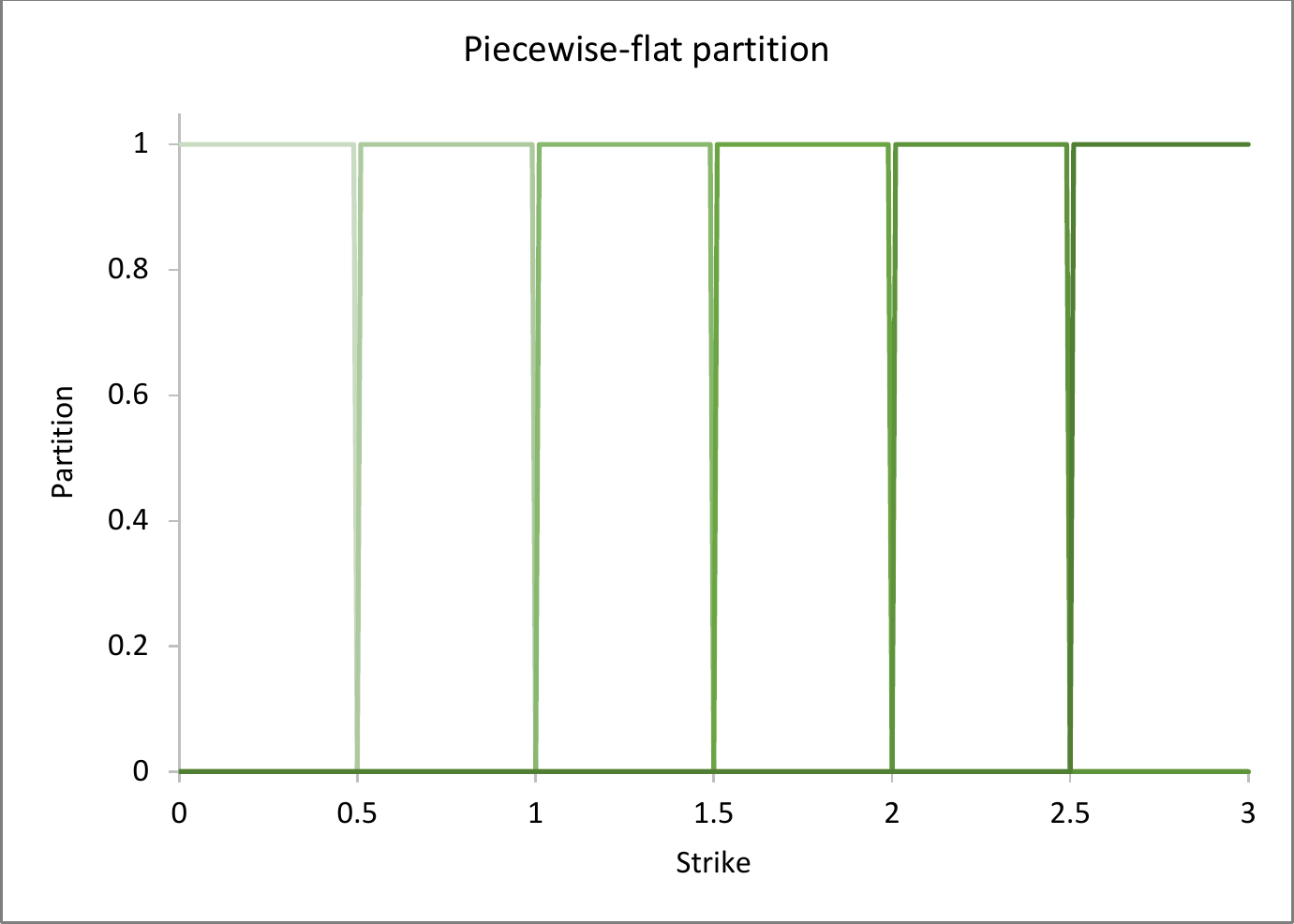}%
\end{tabular}%
\caption{Refining the upper bound for the vanilla option price. The graph
shows the piecewise-flat partition functions constructed with 5 strikes
evenly distributed from 0.5 to 2.5.}
\end{figure}

\begin{figure}[p]
\begin{tabular}{c}
\includegraphics[width=0.95%
\linewidth]{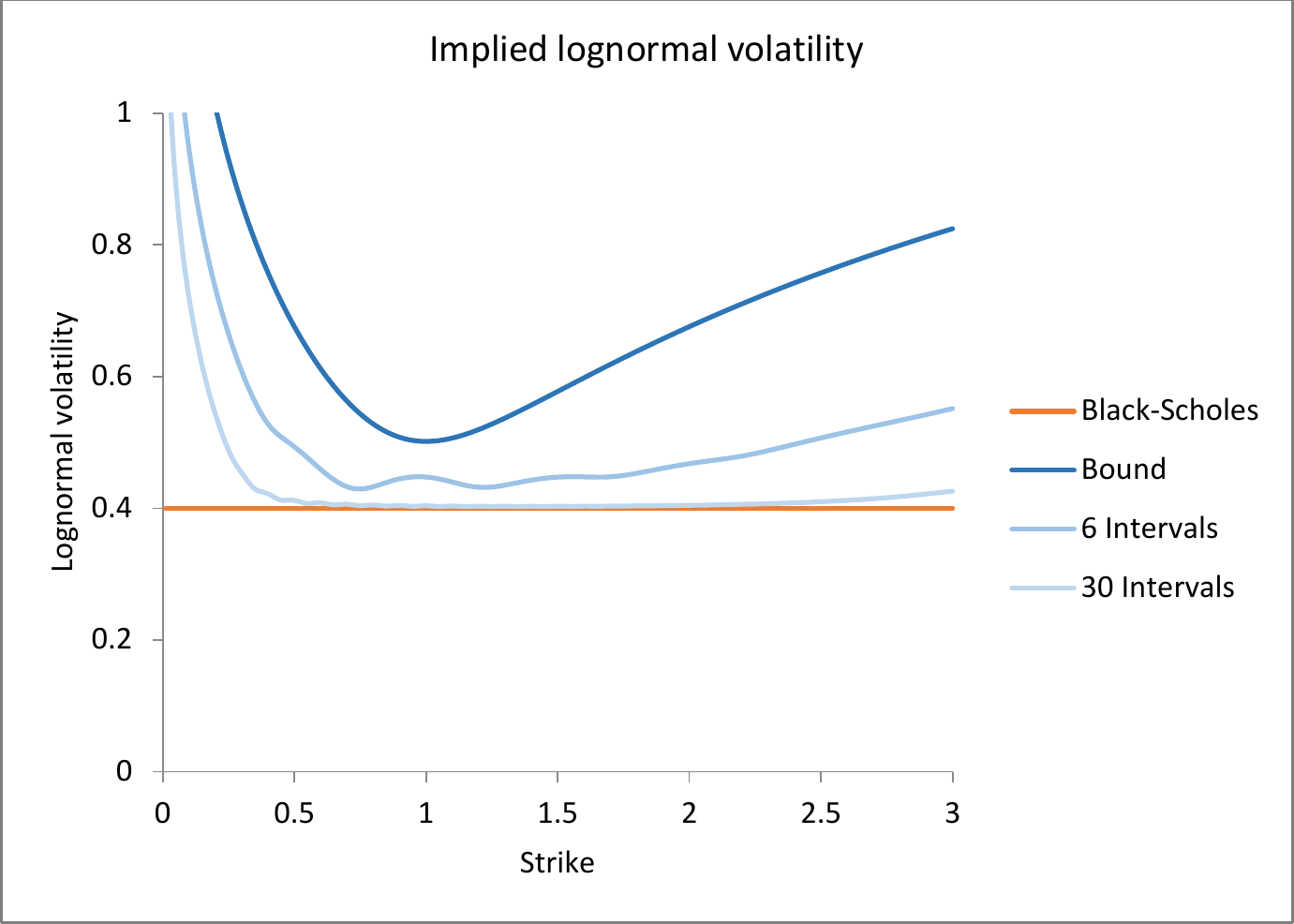}
\\ 
\includegraphics[width=0.95%
\linewidth]{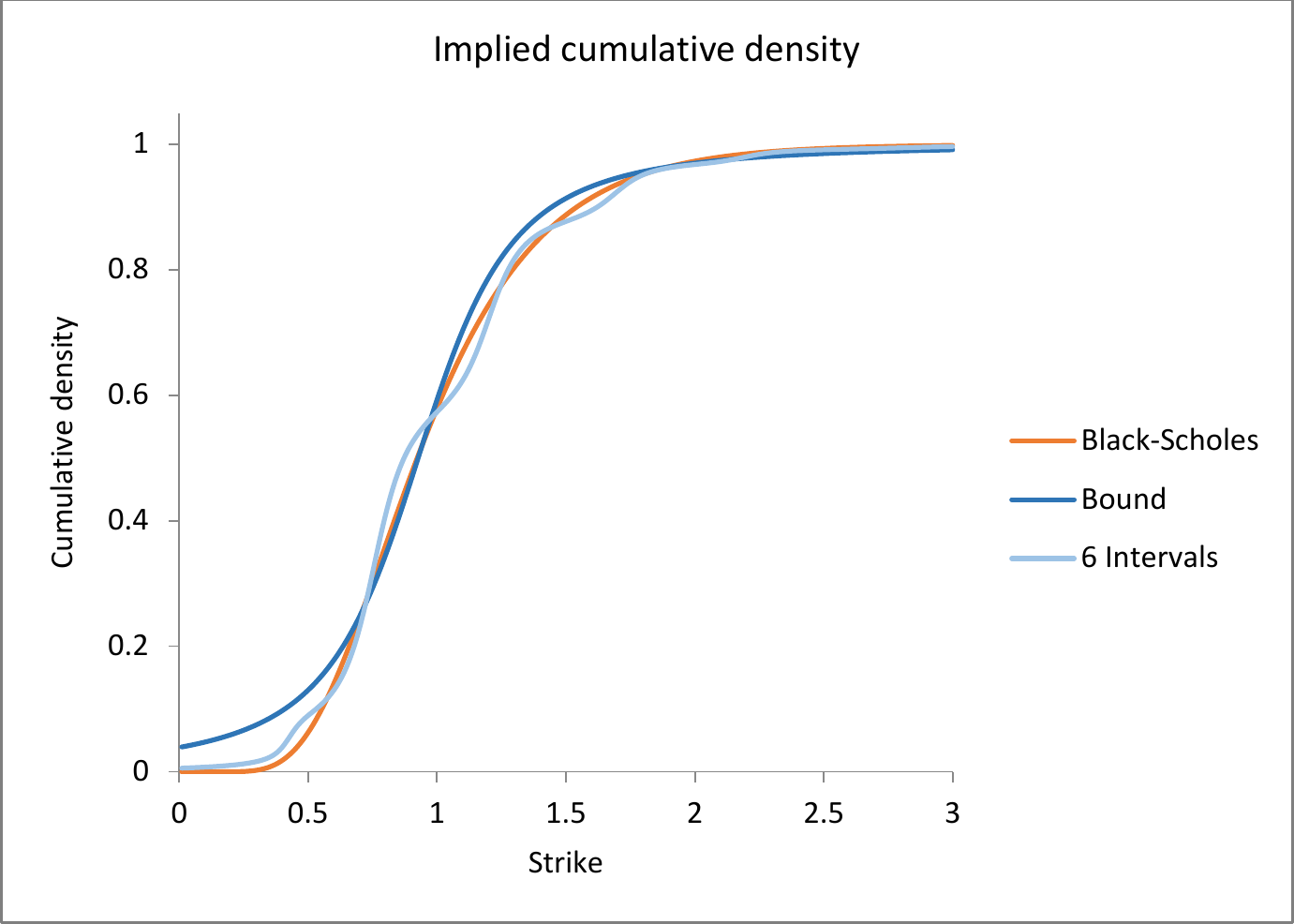}%
\end{tabular}%
\caption{Refining the upper bound for the vanilla option price. In these
graphs, the calibrated moments for the upper bound are extracted from a
Black-Scholes model with mean 1 and volatility 40\%. The three upper bounds
shown correspond to three different subdivisions of the upper half-line,
with 1, 6 and 30 intervals respectively. In the case of 6 intervals, the
boundaries are zero, 5 strikes evenly distributed from 0.5 to 2.5, and
infinity. In the case of 30 intervals, the boundaries are zero, 29 strikes
evenly distributed from 0.1 to 2.9, and infinity. Increasing the number of
intervals adds more information to the upper bound, refining it and
converging towards the Black-Scholes model.}
\end{figure}

A simple example constructs the partition from a decomposition of the upper
half-line into subsets $U_{n}\subset \mathbb{R}_{+}$ satisfying $\cup
_{n}U_{n}=\mathbb{R}_{+}$ and $U_{m}\cap U_{n}=\emptyset $ for $m\neq n$.
The partition comprises the digital options on the asset indicated by the
subsets:%
\begin{equation}
u_{n}[\mathsf{a}]=(\mathsf{a}\in U_{n})
\end{equation}%
with prices $d_{n}$ given by:%
\begin{equation}
d_{n}=\mathbb{E}[(\mathsf{a}\in U_{n})]
\end{equation}%
The matrices $Q$ and $\Lambda $ both divide into four quadrants, with each
quadrant containing a diagonal matrix:%
\begin{align}
Q& =\left[ \begin{array}{c;{1pt/1pt}c} f_{n}d_{n}\delta _{mn} &
\sqrt{f_{n}(1-\nu _{n})}d_{n}\delta _{mn} \\ \hdashline[1pt/1pt]
\sqrt{f_{n}(1-\nu _{n})}d_{n}\delta _{mn} & d_{n}\delta _{mn}\end{array}%
\right] \\
\Lambda & =\left[ \begin{array}{c;{1pt/1pt}c} \delta _{mn} & 0 \\
\hdashline[1pt/1pt] 0 & -k\delta _{mn}\end{array}\right]  \notag
\end{align}%
where $f_{n}$ is the price of the asset and $\nu _{n}$ is the normalised
variance of the square-root of the asset conditional on the asset being in
the subset $U_{n}$:%
\begin{align}
f_{n}& =\mathbb{E}_{n}[\mathsf{a}] \\
\nu _{n}& =\frac{\mathbb{E}_{n}[\mathsf{a}]-\mathbb{E}_{n}[\sqrt{\mathsf{a}}%
]^{2}}{\mathbb{E}_{n}[\mathsf{a}]}  \notag
\end{align}%
In these definitions, the measure $\mathbb{E}_{n}$ is the measure $\mathbb{E}
$ conditional on $(\mathsf{a}\in U_{n})$, defined by:%
\begin{equation}
\mathbb{E}_{n}[\mathsf{b}]=\frac{\mathbb{E}[\mathsf{b}(\mathsf{a}\in U_{n})]%
}{\mathbb{E}[(\mathsf{a}\in U_{n})]}
\end{equation}%
for the security $\mathsf{b}$. The digital prices satisfy $0\leq d_{n}\leq 1$%
, while the conditional price is positive, $f_{n}>0$, and the conditional
root-variance lies in the range $0\leq \nu _{n}\leq 1$. These conditional
moments are normalised by:%
\begin{align}
\sum_{n}d_{n}& =1 \\
\sum_{n}f_{n}d_{n}& =f  \notag \\
\sum_{n}\sqrt{f_{n}(1-\nu _{n})}d_{n}& =\sqrt{f(1-\nu )}  \notag
\end{align}

The decomposition of the upper half-line refines the bound for the option
price, using a breakdown of the asset price and root-variance conditional on
the asset being localised in nominated subsets. As the decomposition is
further refined, more information from the pricing measure is incorporated
into the bound, which converges to the price of the option in the limit of
pointwise localisation.

\subsection{Refinements based on option payoffs}

An alternative approach determines the bound for the option price from the
prices of options at a finite set of strikes. For the positive strikes $%
k_{1}<\cdots <k_{N}$, the partition comprises the functions:%
\begin{align}
u_{1}[\mathsf{a}]& =1-\frac{(\mathsf{a}-k_{1})^{+}-(\mathsf{a}-k_{2})^{+}}{%
k_{2}-k_{1}} \\
u_{n}[\mathsf{a}]& =1-\frac{(k_{n}-\mathsf{a})^{+}-(k_{n-1}-\mathsf{a})^{+}}{%
k_{n}-k_{n-1}}-\frac{(\mathsf{a}-k_{n})^{+}-(\mathsf{a}-k_{n+1})^{+}}{%
k_{n+1}-k_{n}}  \notag \\
u_{N}[\mathsf{a}]& =1-\frac{(k_{N}-\mathsf{a})^{+}-(k_{N-1}-\mathsf{a})^{+}}{%
k_{N}-k_{N-1}}  \notag
\end{align}%
These functions are positive and sum to one, supported on the domains:%
\begin{align}
\func{supp}[u_{1}]& =(-\infty ,k_{2}) \\
\func{supp}[u_{n}]& =(k_{n-1},k_{n+1})  \notag \\
\func{supp}[u_{N}]& =(k_{N-1},\infty )  \notag
\end{align}%
Aside from the diagonal products, only consecutive functions have nonzero
products:%
\begin{equation}
\sqrt{u_{n}[\mathsf{a}]u_{n+1}[\mathsf{a}]}=(k_{n}<\mathsf{a}<k_{n+1})\frac{%
\sqrt{(\mathsf{a}-k_{n})(k_{n+1}-\mathsf{a})}}{k_{n+1}-k_{n}}
\end{equation}

\begin{figure}[t]
\begin{tabular}{c}
\includegraphics[width=0.95%
\linewidth]{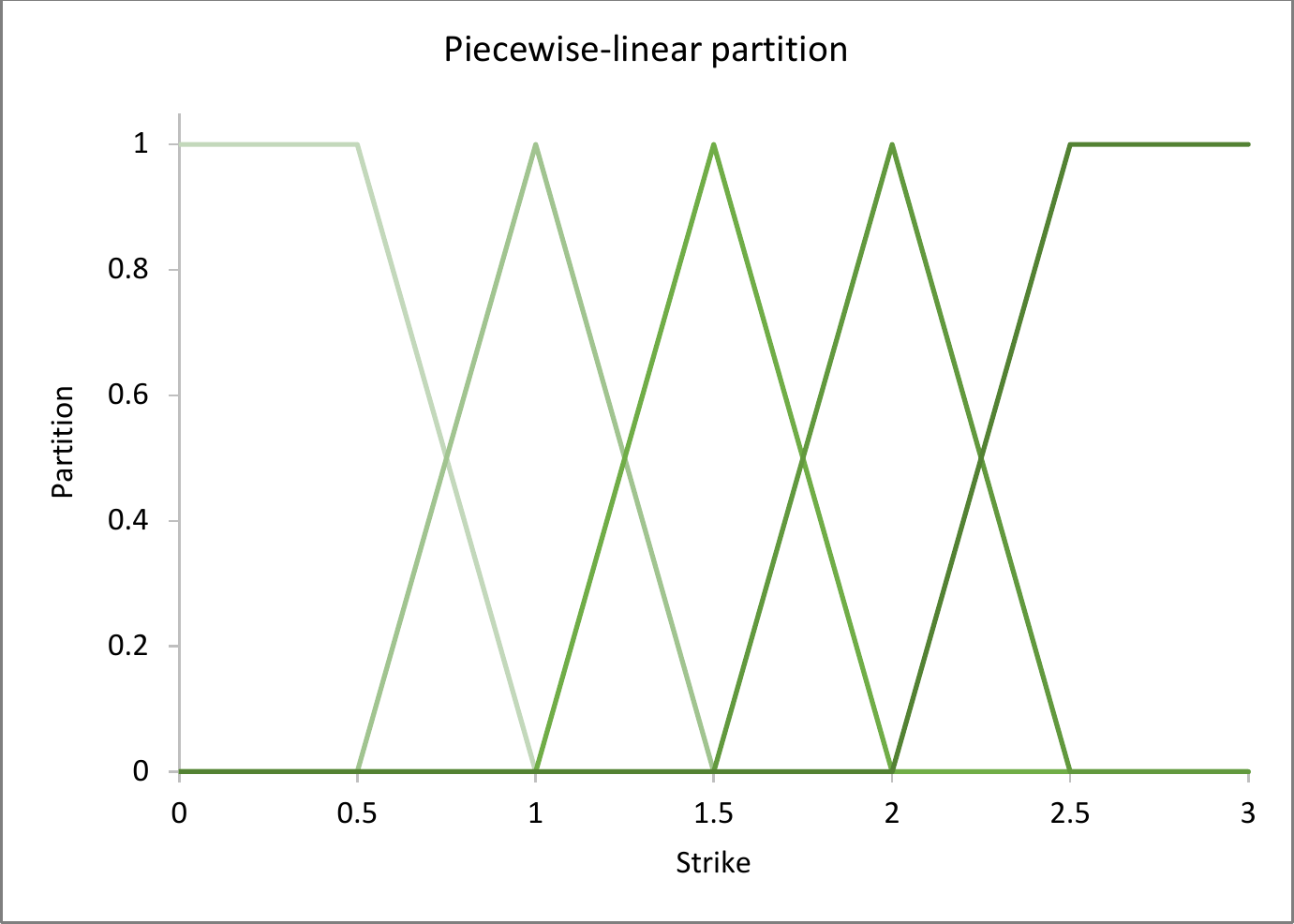}%
\end{tabular}%
\caption{Refining the upper bound for the vanilla option price. The graph
shows the piecewise-linear partition functions constructed with 5 strikes
evenly distributed from 0.5 to 2.5.}
\end{figure}

\begin{figure}[p]
\begin{tabular}{c}
\includegraphics[width=0.95%
\linewidth]{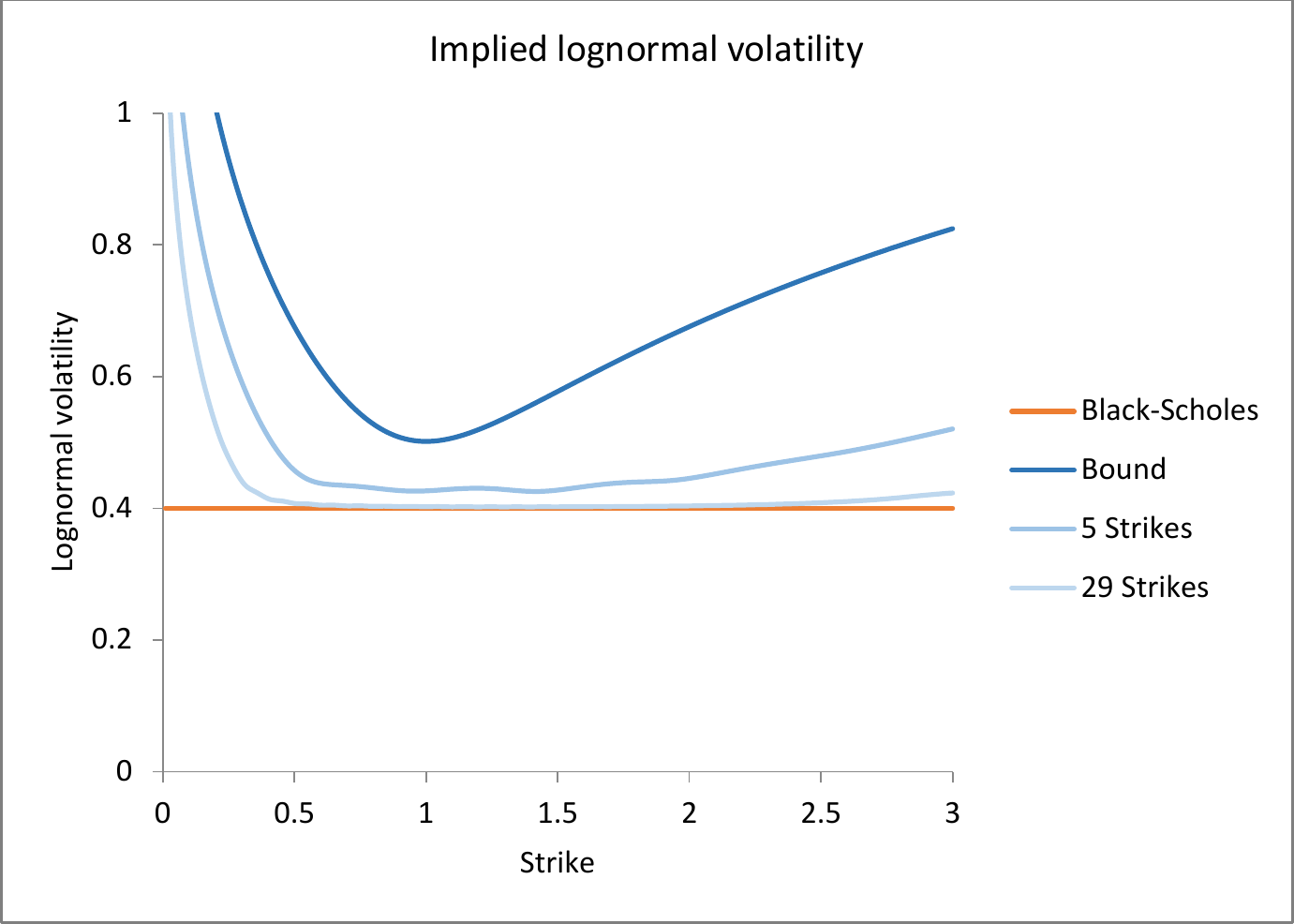}
\\ 
\includegraphics[width=0.95%
\linewidth]{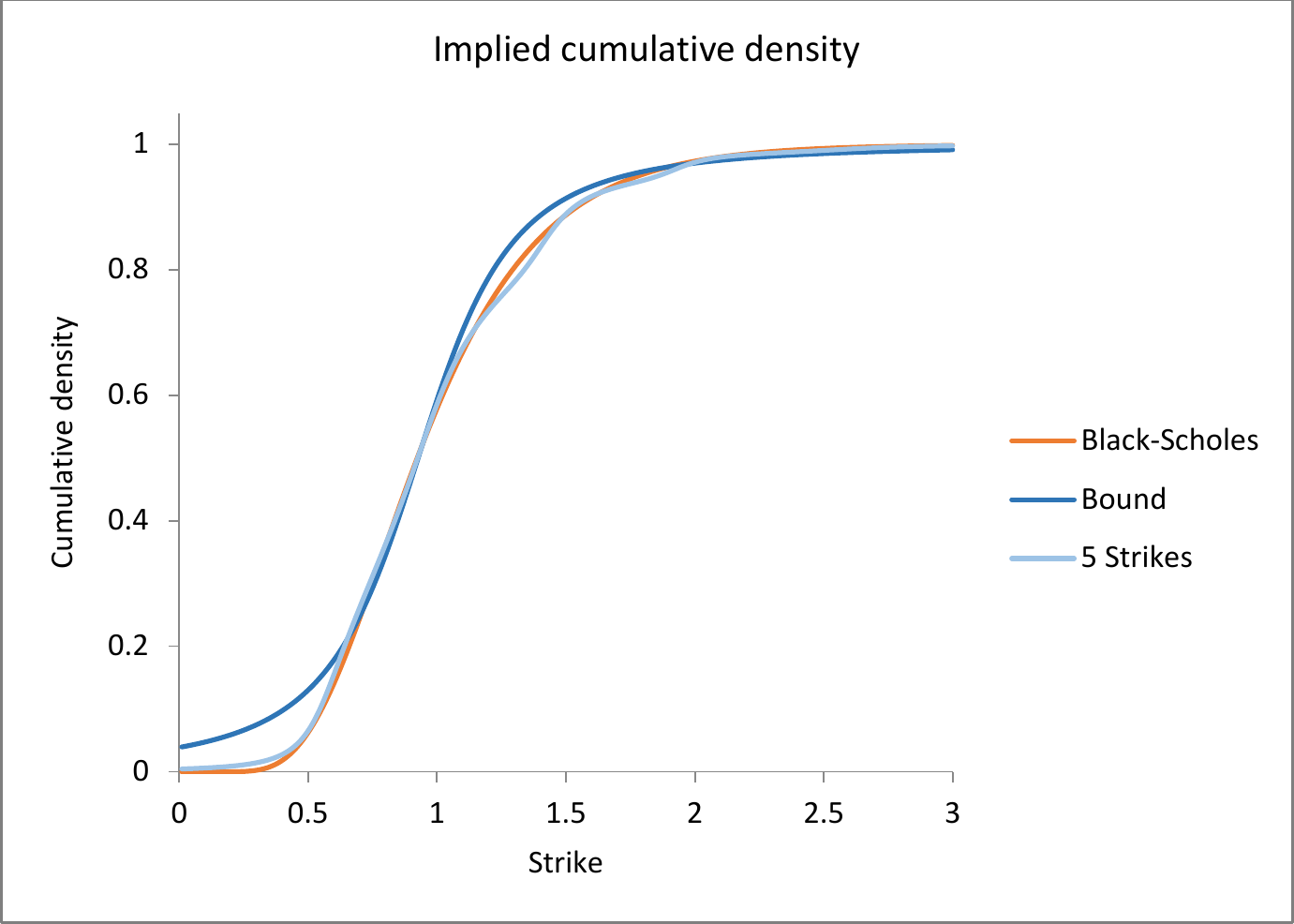}%
\end{tabular}%
\caption{Refining the upper bound for the vanilla option price. In these
graphs, the calibrated moments for the upper bound are extracted from a
Black-Scholes model with mean 1 and volatility 40\%. The three upper bounds
shown correspond to three different sets of strikes, with 0, 5 and 29
strikes respectively. In the case of 5 strikes, the strikes are evenly
distributed from 0.5 to 2.5. In the case of 29 strikes, the strikes are
evenly distributed from 0.1 to 2.9. This example uses piecewise-linear
partition functions, rather than the piecewise-flat partition functions of
the previous example, and this improves the smoothness of the refined bound.}
\end{figure}

The four quadrants of the $2N$-dimensional matrix $Q$ are tridiagonal. The
upper-left quadrant has nonzero elements:%
\begin{align}
Q_{nn}& =\mathbb{E}[\mathsf{a}u_n[\mathsf{a}]] \\
Q_{n\,n+1}=Q_{n+1\,n}& =\mathbb{E}[\mathsf{a}\sqrt{u_n[\mathsf{a}]u_{n+1}[%
\mathsf{a}]}]  \notag
\end{align}%
The upper-right and lower-left quadrants have nonzero elements:%
\begin{align}
Q_{nn}& =\mathbb{E}[\sqrt{\mathsf{a}}u_n[\mathsf{a}]] \\
Q_{n\,n+1}=Q_{n+1\,n}& =\mathbb{E}[\sqrt{\mathsf{a}u_n[\mathsf{a}]u_{n+1}[%
\mathsf{a}]}]  \notag
\end{align}%
The lower-right quadrant has nonzero elements:%
\begin{align}
Q_{nn}& =\mathbb{E}[u_n[\mathsf{a}]] \\
Q_{n\,n+1}=Q_{n+1\,n}& =\mathbb{E}[\sqrt{u_n[\mathsf{a}]u_{n+1}[\mathsf{a}]}]
\notag
\end{align}%
These elements are driven by correlations between the options that are
param\-etrised similarly to the basket option case. In practice, a simple
parametric model, such as the Black-Scholes model, can be used to imply
sensible values for the correlations dependent on a smaller set of model
parameters. By using continuous basis functions, this partition generates a
smoother bound than that implied by the digital partition of the previous
example.

\subsection{Foreign exchange options}

The previous examples show how families of bounds for the option price can
be obtained by subdividing the assets into localised components. An
alternative strategy decomposes the asset into more fundamental economic
units, using an understanding of the financial structure of the asset. The
next example applies this approach to generate bounds for the price of an
option on a cross FX rate.

The option on the FX rate $\mathsf{x}$ is a vanilla option in the form
considered previously, and the option price is subject to the same bound:%
\begin{equation}
\mathbb{E}[(\mathsf{x}-k)^{+}]\leq \frac{1}{2}(f-k)+\frac{1}{2}\sqrt{%
(f-k)^{2}+4fk\nu }
\end{equation}%
where $f$ is the price of the FX rate and $\nu $ is the normalised variance
of the square-root of the FX rate:%
\begin{align}
f& =\mathbb{E}[\mathsf{x}] \\
\nu & =\frac{\mathbb{E}[\mathsf{x}]-\mathbb{E}[\sqrt{\mathsf{x}}]^{2}}{%
\mathbb{E}[\mathsf{x}]}  \notag
\end{align}%
In these expressions, the pricing measure $\mathbb{E}$ is associated with
the base currency of the FX rate.

Options on an illiquid FX rate are more commonly written on the cross FX
rate $\mathsf{x}=\mathsf{x}_{1}/\mathsf{x}_{2}$, where $\mathsf{x}_{1}$ and $%
\mathsf{x}_{2}$ are the FX rates for the two currencies versus a fixed
domestic currency. In this situation, the bound for the vanilla option
continues to apply, but the volatility can be further decomposed in terms of
the volatilities and correlation of the two contributing FX rates:%
\begin{equation}
\nu =1-(\sqrt{(1-\nu _{1})(1-\nu _{2})}+\rho \sqrt{\nu _{1}\nu _{2}})^{2}
\end{equation}%
where $\nu _{1}$ and $\nu _{2}$ are the normalised variances of the
square-roots of the liquid FX rates and $\rho $ is the correlation between
the square-roots of the liquid FX rates:%
\begin{align}
\nu _{1}& =\frac{\bar{\mathbb{E}}[\mathsf{x}_{1}]-\bar{\mathbb{E}}[\sqrt{%
\mathsf{x}_{1}}]^{2}}{\bar{\mathbb{E}}[\mathsf{x}_{1}]} \\
\nu _{2}& =\frac{\bar{\mathbb{E}}[\mathsf{x}_{2}]-\bar{\mathbb{E}}[\sqrt{%
\mathsf{x}_{2}}]^{2}}{\bar{\mathbb{E}}[\mathsf{x}_{2}]}  \notag \\
\rho & =\frac{\bar{\mathbb{E}}[\sqrt{\mathsf{x}_{1}\mathsf{x}_{2}}]-\bar{%
\mathbb{E}}[\sqrt{\mathsf{x}_{1}}]\mathbb{E}[\sqrt{\mathsf{x}_{2}}]}{\sqrt{(%
\bar{\mathbb{E}}[\mathsf{x}_{1}]-\bar{\mathbb{E}}[\sqrt{\mathsf{x}_{1}}%
]^{2})(\bar{\mathbb{E}}[\mathsf{x}_{2}]-\bar{\mathbb{E}}[\sqrt{\mathsf{x}_{2}%
}]^{2})}}  \notag
\end{align}%
In these expressions, the pricing measure $\bar{\mathbb{E}}$ is associated
with the domestic currency, related to the pricing measure $\mathbb{E}$ by:%
\begin{equation}
\mathbb{E}[\mathsf{a}]=\frac{\bar{\mathbb{E}}[\mathsf{ax}_{2}]}{\bar{\mathbb{%
E}}[\mathsf{x}_{2}]}
\end{equation}%
If the two contributing FX rates have liquid option markets, the
volatilities $\nu _{1}$ and $\nu _{2}$ can be replicated from the prices of
options, leaving the correlation $\rho $ to parametrise the bound for the
price of the option on the cross FX rate.

\subsection{Swaptions and caplets}

Another example, taken from the interest rate market, considers the
decomposition of the swap rate into its constituent forward rates. Inverting
the relationship between swap and forward rates leads to bounds on the
prices of forward-starting caplets expressed in terms of the distributions
of the swap rates.

The $n$-period swap rate $\mathsf{s}_{n}$ is decomposed as the weighted
average of the forward rates $\mathsf{r}_{m}$ for $m=1,\ldots ,n$:%
\begin{equation}
\mathsf{s}_{n}=\sum_{m=1}^{n}\frac{p_{m}\delta _{m}}{\sum_{l=1}^{n}p_{l}%
\delta _{l}}\mathsf{r}_{m}
\end{equation}%
In this expression, $p_{n}$ is the discount factor to the $n$th payment date
and $\delta _{n}$ is the daycount fraction for the $n$th accrual period,
where for simplicity the daycount conventions on the fixed and float legs
are assumed to be the same. The weights in this weighted average are
positive and sum to one. In the following discussion these weights are
assumed to be deterministic, an approximation that not only allows the swap
rate to be expressed as a linear combination of the forward rates, but also
avoids complications with differences in the pricing measures associated
with the annuities. These considerations, while significant, are beyond the
scope of the present article.

Inverting the above relationship, the forward rate $\mathsf{r}_{n}$ is
decomposed in terms of the swap rates $\mathsf{s}_{n}$ and $\mathsf{s}_{n-1}$%
:%
\begin{equation}
\mathsf{r}_{n}=(\lambda _{n}+1)\mathsf{s}_{n}-\lambda _{n}\mathsf{s}_{n-1}
\end{equation}%
with weight given by:%
\begin{equation}
\lambda _{n}=\frac{\sum_{m=1}^{n-1}p_{m}\delta _{m}}{p_{n}\delta _{n}}
\end{equation}%
The general result for the bound on the price of a basket option can be
applied to this decomposition. Consider the forward-starting caplet with
strike $k_{n}$ on the $n$th forward rate $\mathsf{r}_{n}$. The payoff for
the caplet decomposes in terms of the swap rates $\mathsf{s}_{n}$ and $%
\mathsf{s}_{n-1}$:%
\begin{equation}
\mathsf{r}_{n}-k_{n}=(\lambda _{n}+1)\mathsf{s}_{n}-\lambda _{n}\mathsf{s}%
_{n-1}-k_{n}
\end{equation}%
In this application, the matrices $Q$ and $\Lambda $ are the
three-dimensional matrices given by:%
\begin{align}
Q& =%
\begin{bmatrix}
f_{n} & \sqrt{f_{n-1}f_{n}}\,q_{n-1\,n} & \sqrt{f_{n}(1-\nu _{n})} \\ 
\sqrt{f_{n-1}f_{n}}\,q_{n-1\,n} & f_{n-1} & \sqrt{f_{n-1}(1-\nu _{n-1})} \\ 
\sqrt{f_{n}(1-\nu _{n})} & \sqrt{f_{n-1}(1-\nu _{n-1})} & 1%
\end{bmatrix}
\\
\Lambda & =%
\begin{bmatrix}
\lambda _{n}+1 & 0 & 0 \\ 
0 & -\lambda _{n} & 0 \\ 
0 & 0 & -k_{n}%
\end{bmatrix}
\notag
\end{align}%
Here, $f_{n}$ is the price of the $n$th swap rate and $\nu _{n}$ is the
normalised variance of the square-root of the $n$th swap rate:%
\begin{align}
f_{n}& =\mathbb{E}[\mathsf{s}_n] \\
\nu _{n}& =\frac{\mathbb{E}[\mathsf{s}_n]-\mathbb{E}[\sqrt{\mathsf{s}_n}]^{2}%
}{\mathbb{E}[\mathsf{s}_n]}  \notag
\end{align}%
and the cross-term $q_{mn}$ is defined from the correlation $\rho _{mn}$
between the square-roots of the $m$th and $n$th swap rates:%
\begin{equation}
\rho _{mn}=\frac{\mathbb{E}[\sqrt{\mathsf{s}_m\mathsf{s}_n}]-\mathbb{E}[%
\sqrt{\mathsf{s}_m}]\mathbb{E}[\sqrt{\mathsf{s}_n}]}{\sqrt{(\mathbb{E}[%
\mathsf{s}_m]-\mathbb{E}[\sqrt{\mathsf{s}_m}]^{2})(\mathbb{E}[\mathsf{s}_n]-%
\mathbb{E}[\sqrt{\mathsf{s}_n}]^{2})}}
\end{equation}%
by:%
\begin{equation}
q_{mn}=\sqrt{(1-\nu _{m})(1-\nu _{n})}+\rho _{mn}\sqrt{\nu _{m}\nu _{n}}
\end{equation}%
The price is positive, $f_{n}>0$, the root-variance lies in the range $0\leq
\nu _{n}\leq 1$, and the correlation lies in the range $-1\leq \rho
_{mn}\leq 1$. The price and root-variance are determined from the market for
swaps and swaptions, leaving the correlation as the model parameter for the
price of the forward-starting caplet.

\begin{figure}[p]
\begin{tabular}{c}
\includegraphics[width=0.95%
\linewidth]{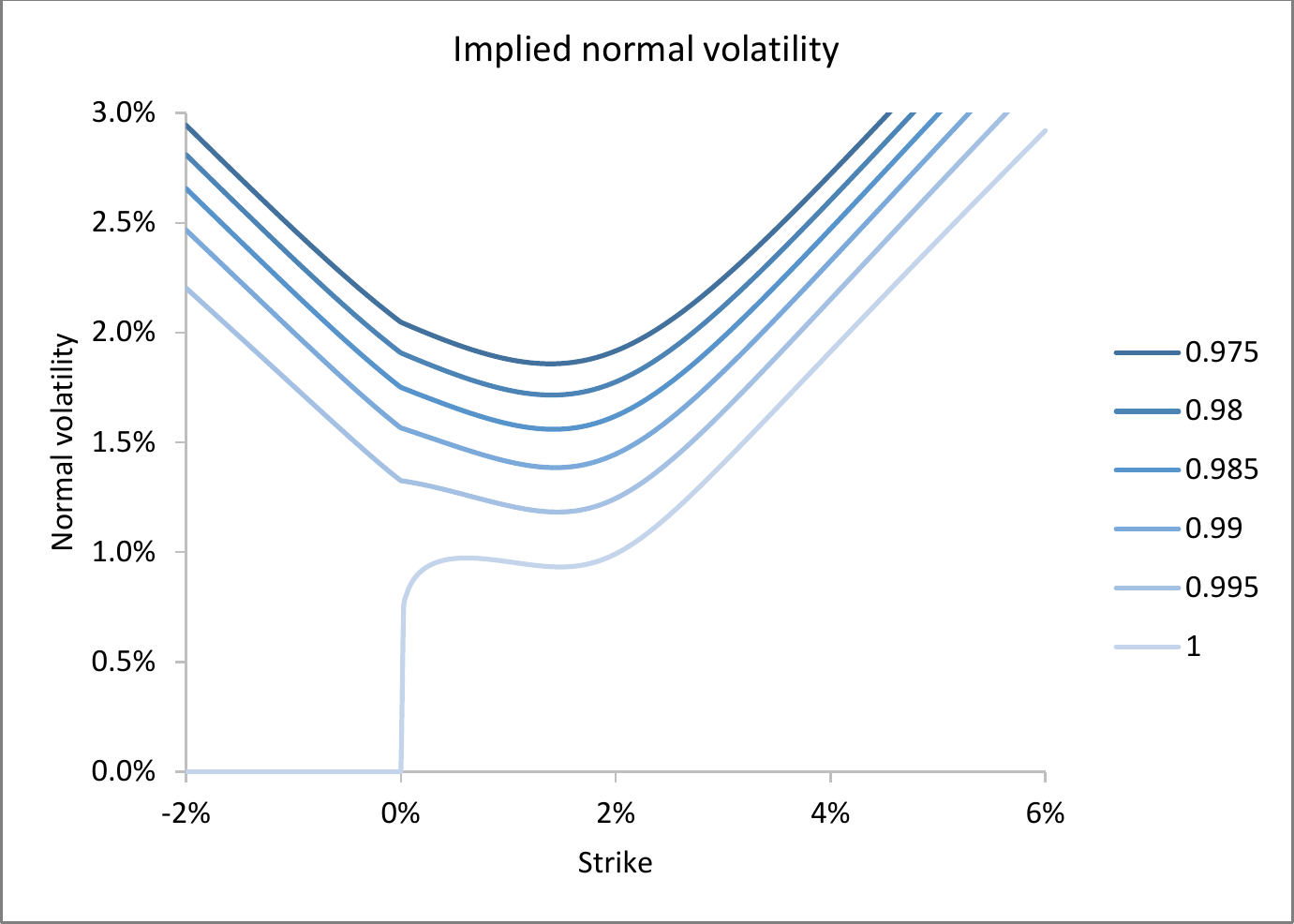}
\\ 
\includegraphics[width=0.95%
\linewidth]{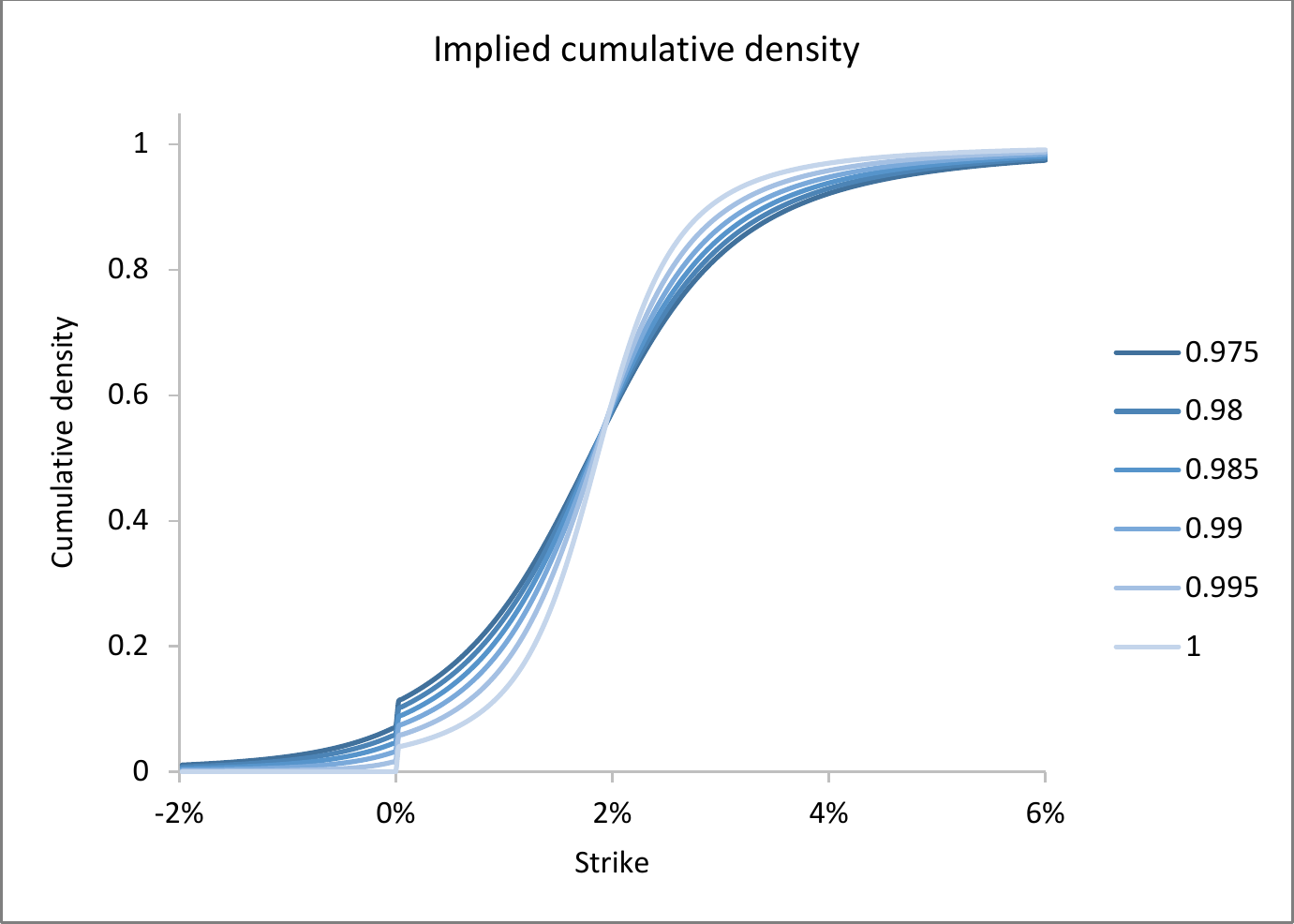}%
\end{tabular}%
\caption{Bounds for the prices of forward-starting caplets. This example
considers the option on the forward rate maturing 10 periods after expiry,
equivalent to an option on the spread between the 10-period and 9-period
swaps. The discount rate is fixed at 1\%, both swap rates are taken to be
2\%, and the root-variances of the swap rates match those generated by the
Black-Scholes model with 40\% volatility. No shift is applied to the swap
rates, and the correlation between the swap rates is varied between 0.975
and 1. The implied cumulative density describes a distribution with
continuous support, with a discrete probability at strike 0\%.}
\end{figure}

\begin{figure}[p]
\begin{tabular}{c}
\includegraphics[width=0.95%
\linewidth]{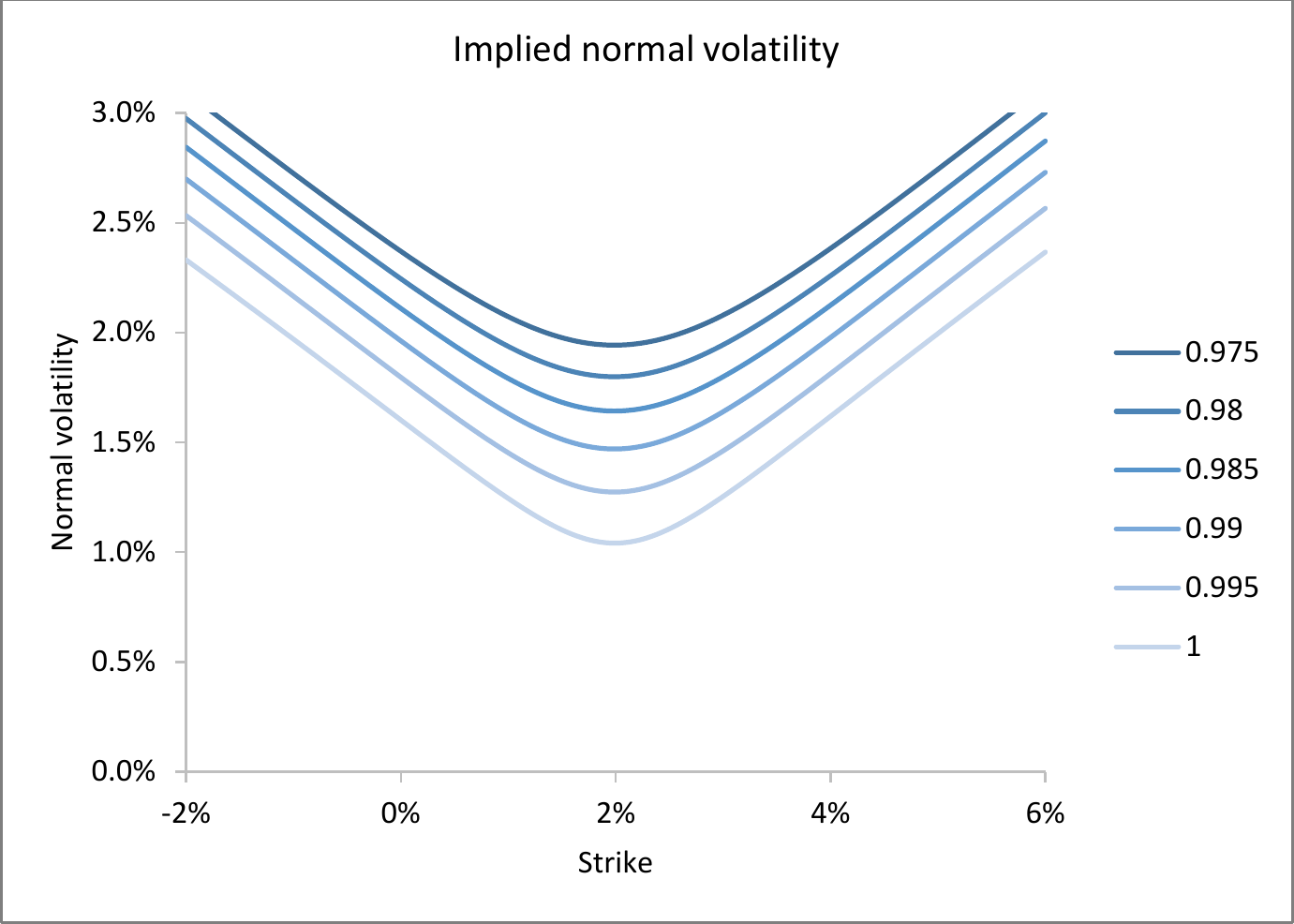}
\\ 
\includegraphics[width=0.95%
\linewidth]{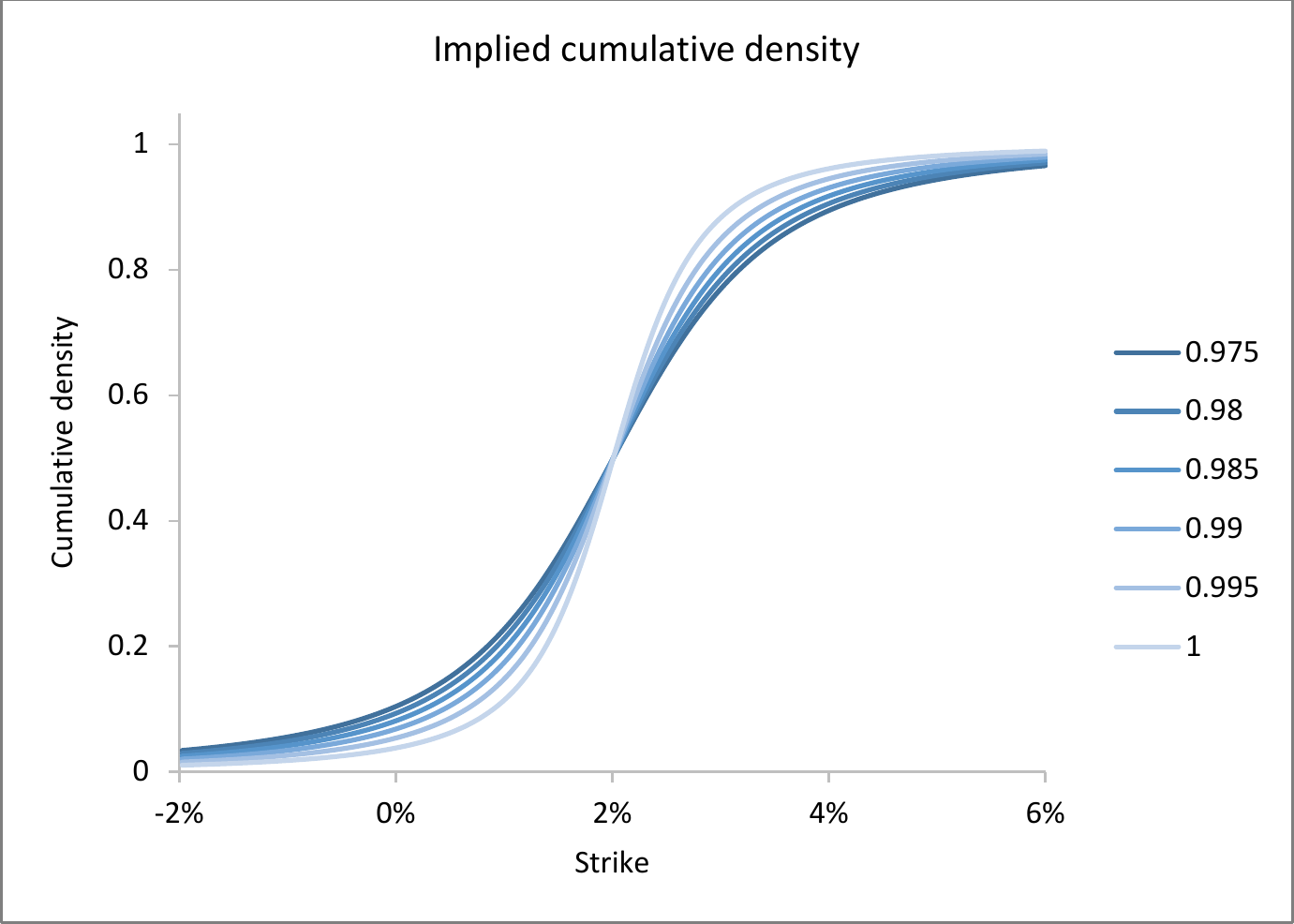}%
\end{tabular}%
\caption{Bounds for the prices of forward-starting caplets. This example
considers the option on the forward rate maturing 10 periods after expiry,
equivalent to an option on the spread between the 10-period and 9-period
swaps. The discount rate is fixed at 1\%, both swap rates are taken to be
2\%, and the root-variances of the swap rates match those generated by the
Black-Scholes model with 40\% volatility. Maximum shift is applied to the
swap rates, and the correlation between the swap rates is varied between
0.975 and 1. The implied cumulative density describes a distribution with
continuous support, without the discrete probability of the previous
example. }
\end{figure}

\begin{figure}[p]
\begin{tabular}{c}
\includegraphics[width=0.95%
\linewidth]{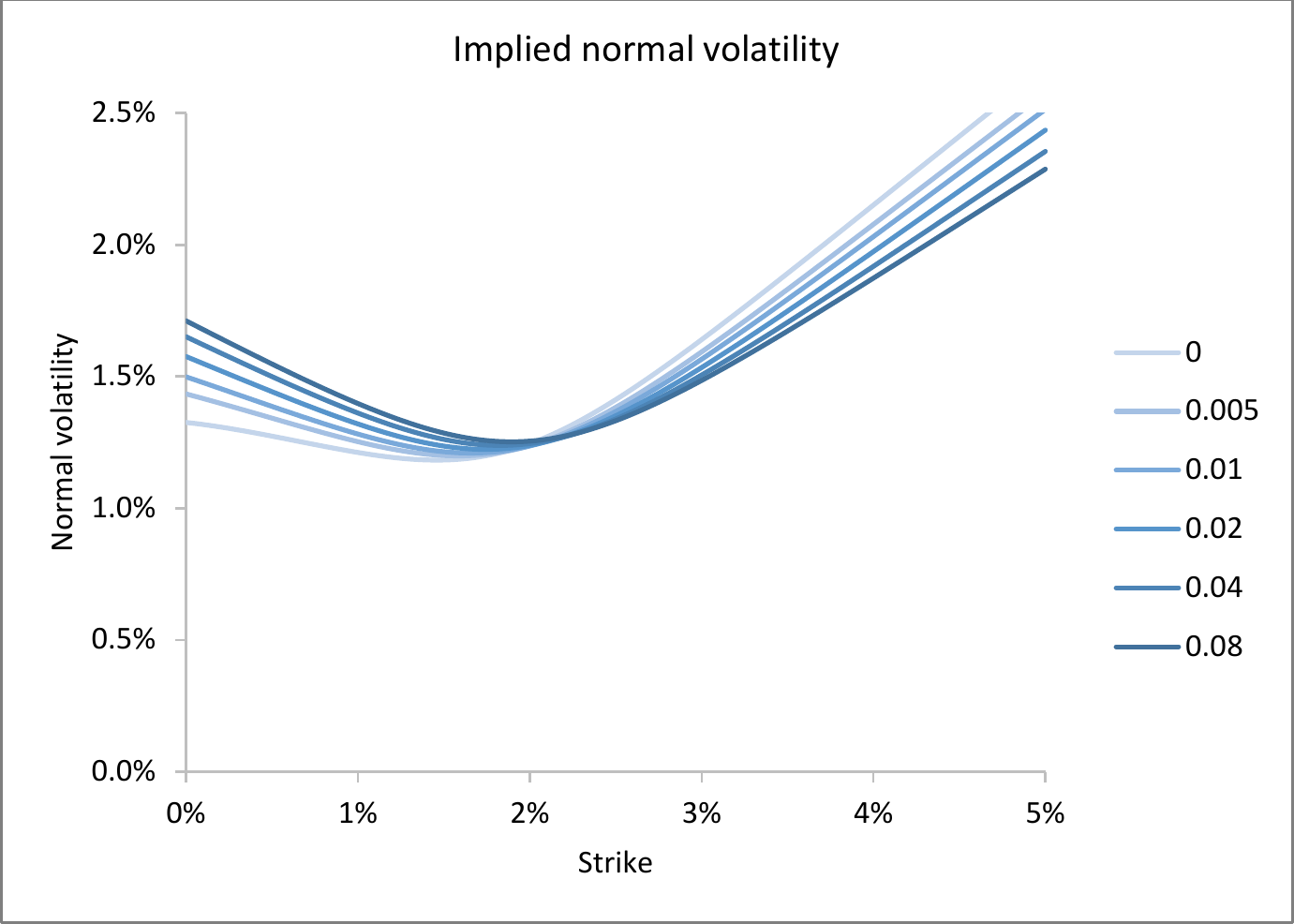}
\\ 
\includegraphics[width=0.95%
\linewidth]{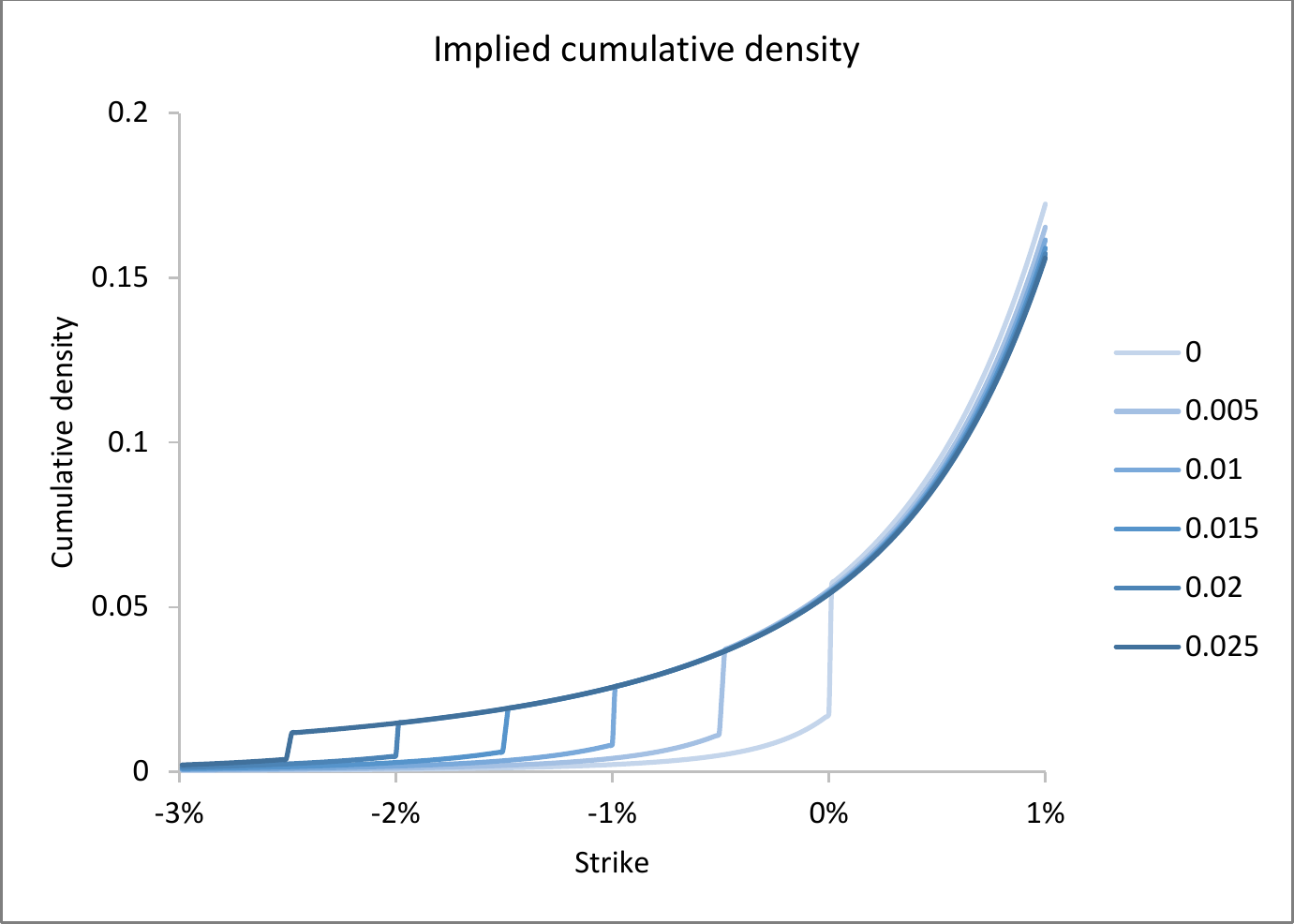}%
\end{tabular}%
\caption{Bounds for the prices of forward-starting caplets. This example
considers the option on the forward rate maturing 10 periods after expiry,
equivalent to an option on the spread between the 10-period and 9-period
swaps. The discount rate is fixed at 1\%, both swap rates are taken to be
2\%, and the root-variances of the swap rates match those generated by the
Black-Scholes model with 40\% volatility. The correlation is fixed at 0.995,
and the shift is varied. The first graph shows how the shift controls the
skew of the implied normal volatility. The second graph shows the tail of
the implied cumulative density, and how the shift moves the position of the
discrete probability. This strike marks the point where the matrix $P$
switches from having two positive eigenvalues to one.}
\end{figure}

There is an implicit assumption that the swap rates are positive in this
application, which cannot be guaranteed. The economic floor for the forward
rate is $\mathsf{r}_{n}\geq -1/\delta _{n}$, so the economic floor for the
swap rate is $\mathsf{s}_{n}\geq -1/\bar{\delta}_{n}$ where the daycount
fraction $\bar{\delta}_{n}$ is the weighted average of the daycount
fractions $\delta _{m}$ for $m=1,\ldots ,n$:%
\begin{equation}
\bar{\delta}_{n}=\sum_{m=1}^{n}\frac{p_{m}}{\sum_{l=1}^{n}p_{l}}\delta _{m}
\end{equation}%
Positivity can be restored to the terms in the decomposition of the forward
rate by shifting the swap rates and strike by a proportion $\alpha $ of the
economic floor:%
\begin{align}
\mathsf{s}_{n}& \mapsto \mathsf{s}_{n}+\alpha /\bar{\delta}_{n} \\
k_{n}& \mapsto k_{n}+\alpha /\delta _{n}  \notag
\end{align}%
Applying these substitutions in the expressions above for the matrices $Q$
and $\Lambda $ generates an upper bound based on swap rates with negative
floors. The additional model parameter in this construction is the shift,
which takes values in the range $0\leq \alpha \leq 1$.

\section{Attaining the option price bound}

Application of the GNS\ construction to the pricing measure generates upper
bounds for the option price, and with the creative decomposition of the
option portfolio this leads to a diverse range of bounds depending on
partial information extracted from the measure. There is, however, no
guarantee that the bound derived from this construction is useful, though
the examples of the previous section suggest this is the case.

One question to ask is whether there is a measure satisfying the constraints
that attains the bound for the option price, and in this statement there are
two variations: local and global attainment. For options on portfolios
generated from the assets $\mathsf{a}_{n}$, the bound is derived from the
matrix with elements $Q_{mn}$. The portfolio quantities are provided by the
scalars $\lambda _{n}$, and the GNS construction generates an option price
bound $p[\lambda ]$ as a function of these quantities. For this
configuration, the local and global attainment of the bound is expressed in
the following definitions.

\begin{description}
\item[Local attainment:] For each portfolio $\lambda $ there is an
arbitrage-free pricing model $\mathbb{E}_{\lambda }$ that satisfies the
constraints:%
\begin{equation}
\mathbb{E}_{\lambda }[\sqrt{\mathsf{a}_{m}\mathsf{a}_{n}}]=Q_{mn}
\end{equation}%
and has option price given by:%
\begin{equation}
\mathbb{E}_{\lambda }[(\sum_{n}\lambda _{n}\mathsf{a}_{n})^{+}]=p[\lambda ]
\end{equation}

\item[Global attainment:] There is an arbitrage-free pricing model $\mathbb{E%
}$ that satisfies the constraints:%
\begin{equation}
\mathbb{E}[\sqrt{\mathsf{a}_{m}\mathsf{a}_{n}}]=Q_{mn}
\end{equation}%
and for each portfolio $\lambda $ has option price given by:%
\begin{equation}
\mathbb{E}[(\sum_{n}\lambda _{n}\mathsf{a}_{n})^{+}]=p[\lambda ]
\end{equation}
\end{description}

The remainder of this article investigates local and global attainment in
the simple case of two assets. Consider the option to exchange the asset $%
\mathsf{a}$ for $k$ units of the asset $1$. The GNS construction provides an
upper bound for the option price across all pricing models $\mathbb{E}$
constrained to match the price $f$ and root-variance $\nu $:%
\begin{equation}
\mathbb{E}[(\mathsf{a}-k)^{+}]\leq \frac{1}{2}(f-k)+\frac{1}{2}\sqrt{%
(f-k)^{2}+4fk\nu }
\end{equation}%
where:%
\begin{align}
f& =\mathbb{E}[\mathsf{a}] \\
\nu & =\frac{\mathbb{E}[\mathsf{a}]-\mathbb{E}[\sqrt{\mathsf{a}}]^{2}}{%
\mathbb{E}[\mathsf{a}]}  \notag
\end{align}%
This bound is attained by the binomial model, albeit with a configuration
that depends on the strike, and this demonstrates local attainment. The
Carr-Madan replication formula shows that the measure implied by the bound
does not match the required moments -- there is no single measure that
generates the bound for all strikes -- so the bound is not globally attained.

\subsection{Local attainment}

In the binomial model, the asset $\mathsf{a}$ with binomial spectrum $\func{%
spec}[\mathsf{a}]=\{a_{-},a_{+}\}\subset \mathbb{R}_{+}$ has price:%
\begin{equation}
\mathbb{E}[\mathsf{a}]=a_{-}\sin [\chi ]^{2}+a_{+}\cos [\chi ]^{2}
\end{equation}%
where the angle $\chi $ in the range $0\leq \chi \leq \pi /2$ generates
positive weights that sum to one. The calibration problem is transformed
into trigonometry by assigning $\nu =\cos [\theta ]^{2}$ for the angle $%
\theta $ in the range $0\leq \theta \leq \pi /2$. Calibration to the price $%
f $ and root-variance $\nu $ then leads to the constraints:%
\begin{align}
\sqrt{f}\sin [\theta ]& =\sqrt{a_{-}}\sin [\chi ]^{2}+\sqrt{a_{+}}\cos [\chi
]^{2} \\
f& =a_{-}\sin [\chi ]^{2}+a_{+}\cos [\chi ]^{2}  \notag
\end{align}

\begin{figure}[p]
\begin{tabular}{c}
\includegraphics[width=0.95%
\linewidth]{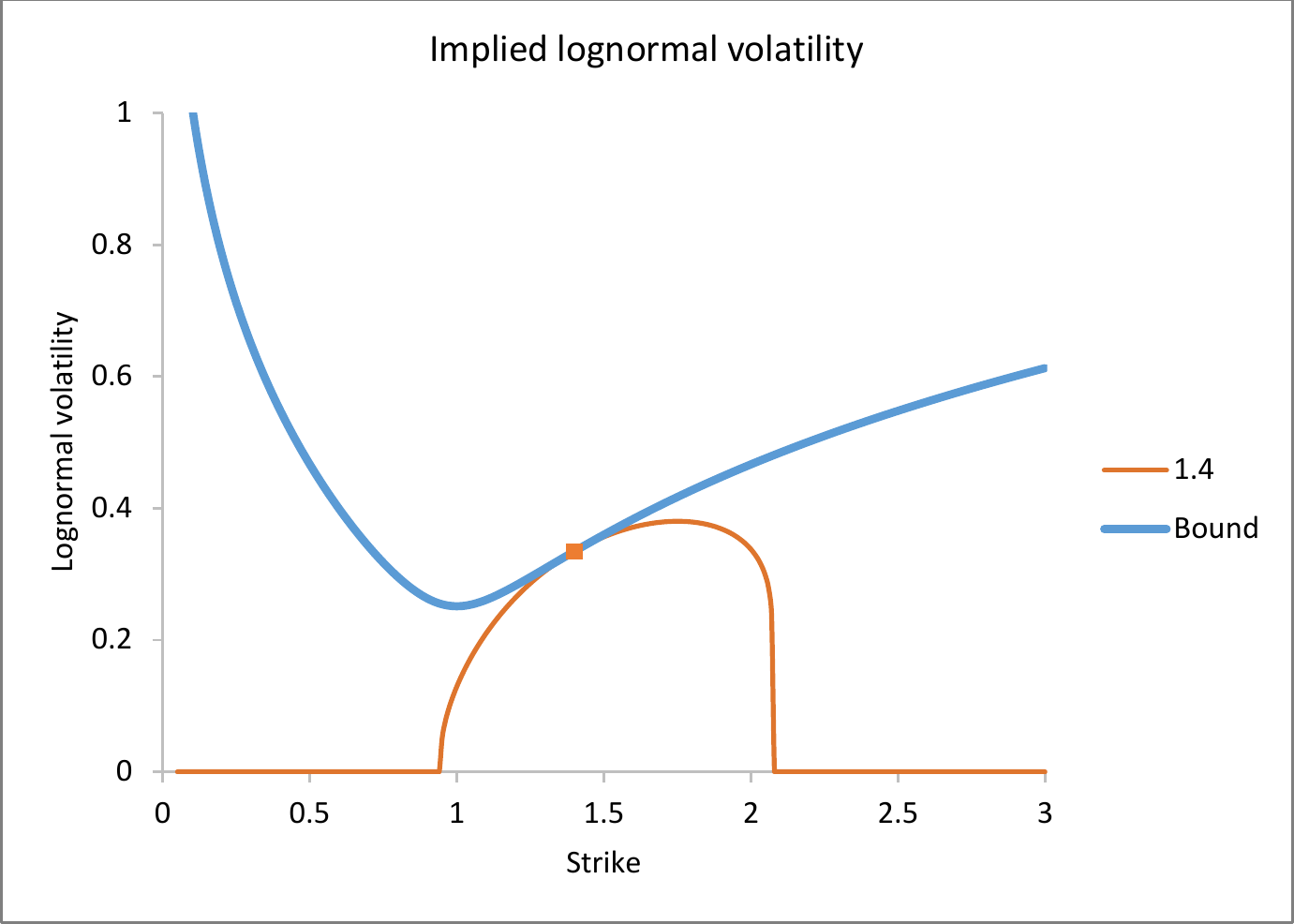}
\\ 
\includegraphics[width=0.95%
\linewidth]{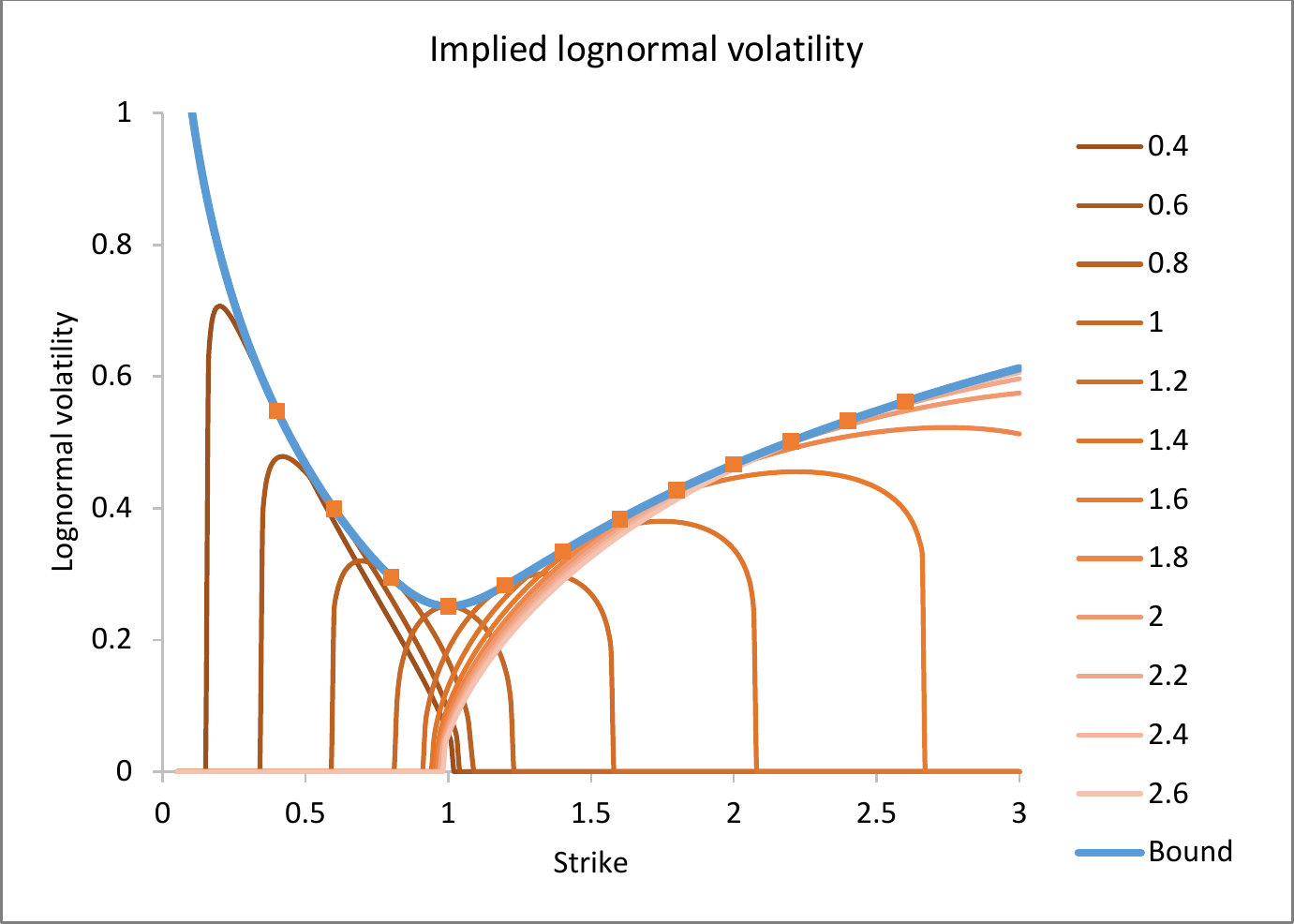}%
\end{tabular}%
\caption{The maximum vanilla option price in the binomial model, compared to
the upper bound. In these graphs, the price of the asset is fixed at 1 and
the root-variance is fixed at 0.01. The first graph shows the implied
lognormal volatility for the binomial model that generates the maximum
option price that can be attained at strike 1.4. The second graph includes
the binomial models generating the maximum option prices that can be
attained at a range of strikes between 0.4 and 2.6. The optimal binomial
model depends on the strike, and the maximum across all these binomial
models matches the upper bound.}
\end{figure}

For a given angle $\chi $, assumed not to be equal to the edge cases $0$ or $%
\pi /2$, these relations can be inverted to identify the asset with the
specified moments. The constraint imposed by calibration to the price is
solved by:%
\begin{align}
a_{-}& =f\frac{\cos [\beta ]^{2}}{\sin [\chi ]^{2}} \\
a_{+}& =f\frac{\sin [\beta ]^{2}}{\cos [\chi ]^{2}}  \notag
\end{align}%
for an angle $\beta $ in the range $0\leq \beta \leq \pi /2$. The constraint
imposed by calibration to the root-variance is then solved for the angle $%
\beta $:%
\begin{equation}
\sin [\theta ]=\sin [\chi +\beta ]
\end{equation}%
There are two solutions to this equation. The first solution $\beta =\theta
-\chi $ is valid for angle $\chi $ in the range $0<\chi \leq \theta $,
leading to the following spectrum for the asset:%
\begin{align}
a_{-}& =f\frac{\cos [\theta -\chi ]^{2}}{\sin [\chi ]^{2}} \\
a_{+}& =f\frac{\sin [\theta -\chi ]^{2}}{\cos [\chi ]^{2}}  \notag
\end{align}%
This solution satisfies $a_{-}\geq a_{+}$. The second solution $\beta =\pi
-\theta -\chi $ is valid for angle $\chi $ in the range $\pi /2-\theta \leq
\chi <\pi /2$, leading to the following spectrum for the asset:%
\begin{align}
a_{-}& =f\frac{\cos [\theta +\chi ]^{2}}{\sin [\chi ]^{2}} \\
a_{+}& =f\frac{\sin [\theta +\chi ]^{2}}{\cos [\chi ]^{2}}  \notag
\end{align}%
This solution satisfies $a_{-}\leq a_{+}$. The two solutions transform into
each other under the transformation $\chi \mapsto \pi /2-\chi $ that
switches the underlying states.

Focussing, without loss of generality, on the second solution, the option
price is maximised at the angle $\chi $ satisfying:%
\begin{equation}
\tan [2\chi ]=-\frac{f\sin [2\theta ]}{f\cos [2\theta ]+k}
\end{equation}%
At this angle, the price of the option is given by the supremum price:%
\begin{equation}
\mathbb{E}[(\mathsf{a}-k)^+]=\frac{1}{2}(f-k)+\frac{1}{2}\sqrt{%
(f-k)^{2}+4fk\nu }
\end{equation}%
The binomial model at this angle generates the supremum option price for
pricing models that calibrate to the asset price and root-variance. This is
not entirely surprising, as the supremum problem is essentially a linear
programming problem, and with two constraints the solution reduces to a
domain comprised of just two states. Note, however, that the angle that
specifies the optimal binomial model depends on the strike. There is no
single binomial model that achieves the bound for all strikes.

\subsection{Global attainment}

The bound for the option price is decreasing and convex as a function of the
strike, and so represents a pricing measure that is free of arbitrage. For
any individual strike, the bound provides the maximum possible option price
from pricing models matching the asset price and root-variance. This does
not imply that the bound itself defines a pricing model that matches the
asset price and root-variance. Application of the Carr-Madan replication
formula demonstrates that the implied measure has root-variance that exceeds
the calibration constraint.

\begin{figure}[p]
\begin{tabular}{c}
\includegraphics[width=0.95%
\linewidth]{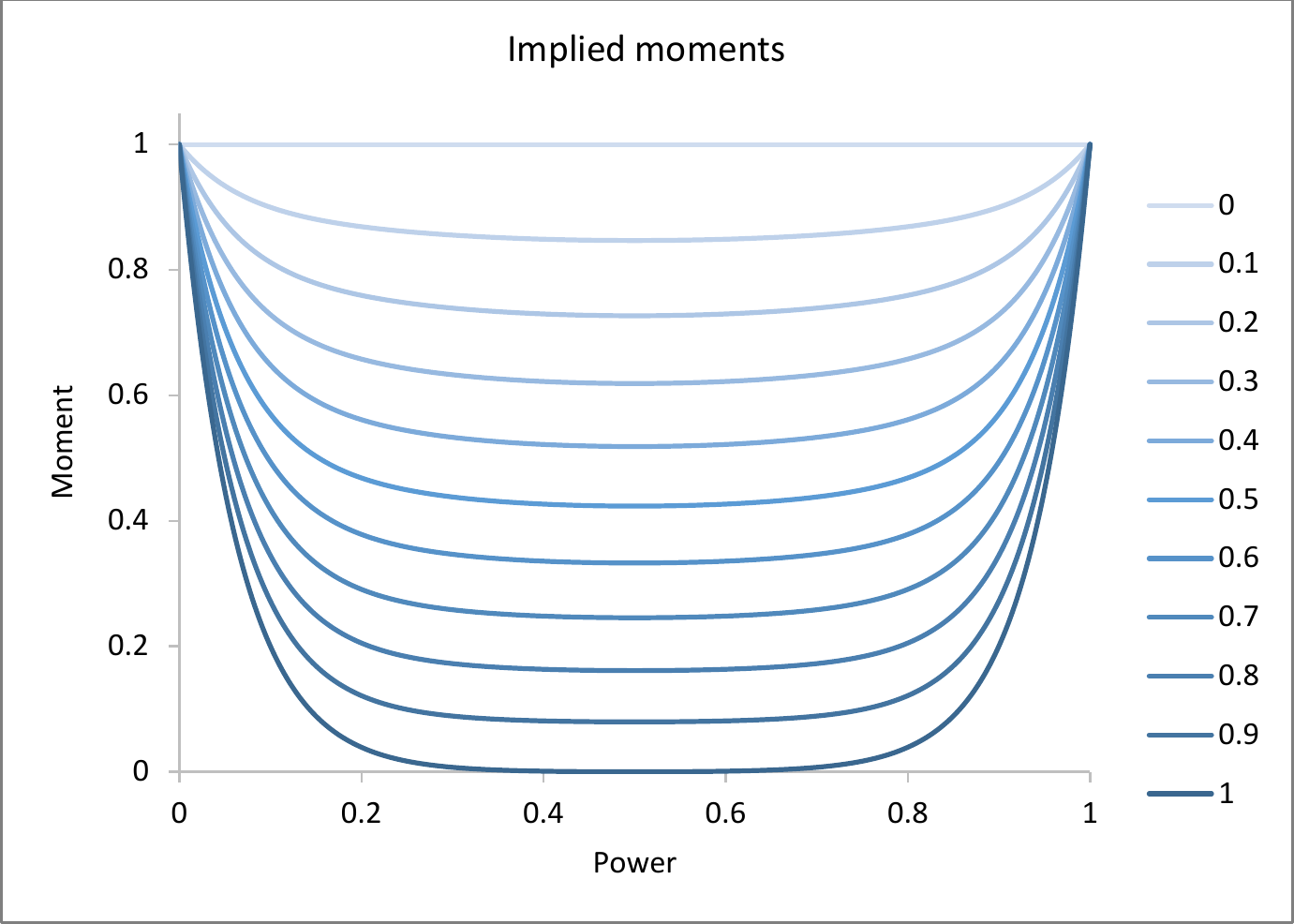}
\\ 
\includegraphics[width=0.95%
\linewidth]{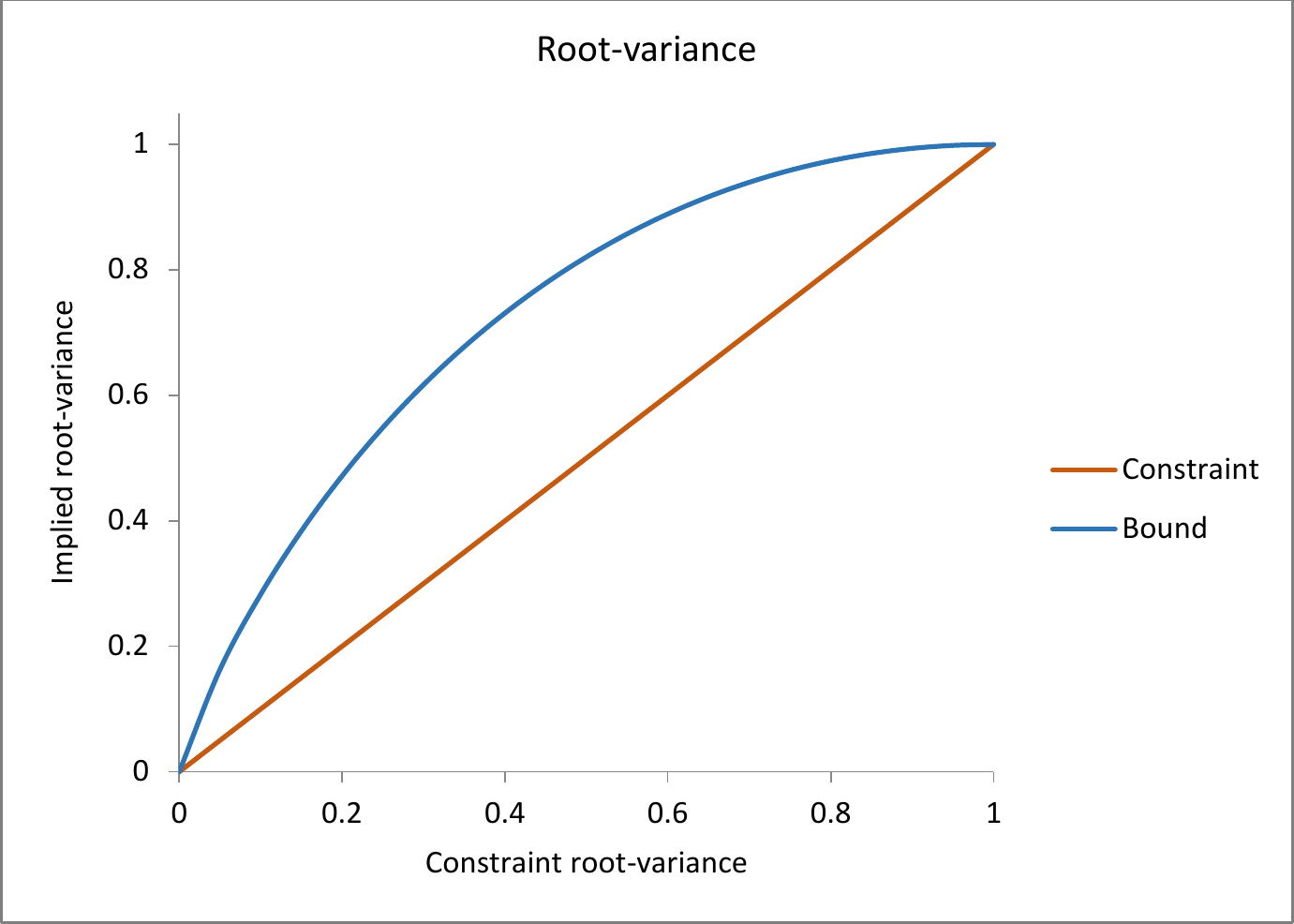}%
\end{tabular}%
\caption{The moments and root-variance implied by the option price bound as
a function of the constraint for the root-variance. The first graph shows
the moments for a range of values for the constraint between 0 and 1. The
second graph then compares the implied root-variance with the constraint
root-variance. The moment is computed using the Carr-Madan replication
formula. Except for the boundary points, the implied root-variance is always
strictly higher than the constraint root-variance.}
\end{figure}

Consider the pricing model $\mathbb{E}$ with call and put option prices
given by:%
\begin{align}
\mathbb{E}[(\mathsf{a}-k)^{+}]& =\frac{1}{2}(f-k)+\frac{1}{2}\sqrt{%
(f-k)^{2}+4fk\nu } \\
\mathbb{E}[(k-\mathsf{a})^{+}]& =\frac{1}{2}(k-f)+\frac{1}{2}\sqrt{%
(k-f)^{2}+4kf\nu }  \notag
\end{align}%
Subtracting these expressions, it immediately follows that the measure is
calibrated to the price $f$:%
\begin{equation}
\mathbb{E}[\mathsf{a}]=f
\end{equation}%
The Carr-Madan replication formula determines the price $\mathbb{E}[\phi
\lbrack \mathsf{a}]]$ of the payoff $\phi \lbrack \mathsf{a}]$ to be:%
\begin{align}
\mathbb{E}[\phi \lbrack \mathsf{a}]]=\phi \lbrack f]& +\frac{1}{2}%
\int_{k=0}^{f}\phi ^{\prime \prime }[k]((k-f)+\sqrt{(k-f)^{2}+4kf\nu })\,dk
\\
& +\frac{1}{2}\int_{k=f}^{\infty }\phi ^{\prime \prime }[k]((f-k)+\sqrt{%
(f-k)^{2}+4fk\nu })\,dk  \notag
\end{align}%
The first integral is simplified with the change of variables $x=\sqrt{k/f}$
and the second integral is simplified with the change of variables $x=\sqrt{%
f/k}$, leading to:%
\begin{align}
\mathbb{E}[\phi \lbrack \mathsf{a}]]=\phi \lbrack f]+f^{2}\int_{x=0}^{1}& 
\frac{1}{x^{2}}(x^{3}\phi ^{\prime \prime }[fx^{2}]+x^{-3}\phi ^{\prime
\prime }[fx^{-2}]) \\
& \times (\sqrt{(1-x^{2})^{2}+4x^{2}\nu }-(1-x^{2}))\,dx  \notag
\end{align}

The moment $\mathbb{E}[\mathsf{a}^{n}]$ is finite only for $0\leq n\leq 1$.
The integer moments are $\mathbb{E}[1]=1$ and $\mathbb{E}[\mathsf{a}]=f$,
and for $0<n<1$ the moment is given by:%
\begin{align}
\frac{\mathbb{E}[\mathsf{a}^{n}]}{f^{n}}=1+n(n-1)\int_{x=0}^{1}& \frac{1}{%
x^{2}}(x^{2n-1}+x^{-(2n-1)}) \\
& \times (\sqrt{(1-x^{2})^{2}+4x^{2}\nu }-(1-x^{2}))\,dx  \notag
\end{align}%
The moment is symmetric under the transformation $n\mapsto 1-n$. The case $%
n=1/2$ corresponds to the fixed point of this transformation:%
\begin{equation}
\frac{\mathbb{E}[\sqrt{\mathsf{a}}]}{\sqrt{f}}=1-\frac{1}{2}\int_{x=0}^{1}%
\frac{1}{x^{2}}(\sqrt{(1-x^{2})^{2}+4x^{2}\nu }-(1-x^{2}))\,dx
\end{equation}%
This expression is numerically integrated to generate the root-variance
implied by the option price bounds. The implied root-variance exceeds the
constraint everywhere except at the edge cases, demonstrating that the bound
is not globally attained.

\section{Conclusion}

By exploring the exercise strategies that are available in the larger
algebra of all operators on the Hilbert space in the GNS construction, the
approach developed here generates bounds for option pricing contingent only
on partial information from the pricing measure. In some cases this is a
tight bound for the option price, being attained by the multinomial model
calibrated to the same target moments, and can be arbitrarily refined by
extracting more information. The family of bounds generated by this approach
depends on the partition of the option portfolio, and with ingenuity leads
to methods for interpolating the volatility smile, linking swaption and
caplet prices, and many other financial applications. Intriguingly, the
volatility smiles implied by these bounds are similar to smiles observed in
the market.

These results accommodate an extension to the classical theory of
mathematical finance that, by admitting noncommuting assets, is amenable to
the methods of quantum analysis. At opposing extremes in this picture are
the classical algebra of left-multiplication operators and the quantum
algebra of all operators on the Hilbert space. There are many layers of
algebra between these extremes, each of which determines a domain for the
exercise strategies, thereby creating a hierarchy of option pricing bounds.
This suggests a relationship between the theory of von Neumann algebras and
the pricing of options that is worthy of further investigation.

\end{document}